\documentclass[a4paper,twocolumn,11pt,accepted=2022-01-06]{quantumarticle}
\pdfoutput=1
\usepackage[utf8]{inputenc}
\usepackage[english]{babel}
\usepackage[T1]{fontenc}
\usepackage{amsmath}

\usepackage{tikz}
\usepackage{lipsum}

\makeatletter
\def\@hangfrom@section#1#2#3{\@hangfrom{#1#2}#3}
\def\@hangfroms@section#1#2{#1#2}
\makeatother

\renewcommand{\thesection}{\arabic{section}} \renewcommand{\thesubsection}{\thesection.\arabic{subsection}} \renewcommand{\thesubsubsection}{\thesubsection.\arabic{subsubsection}}

\usepackage{graphicx}
\usepackage{dcolumn}
\usepackage{bm}
\usepackage{amssymb}
\usepackage{subfigure}
\usepackage{epstopdf}
\usepackage{url}
\usepackage[breaklinks]{hyperref}
\usepackage{breakurl}
\usepackage{rotating}
\usepackage{makecell} 

\usepackage{adjustbox}
\usetikzlibrary{quantikz}
\usepackage[numbers,sort&compress]{natbib}

\usepackage{braket}
\usepackage{ragged2e}
\usepackage{bbold}
\usepackage{xcolor}
\usepackage{mathrsfs}
\usepackage[ruled,noline]{algorithm2e}

\usepackage{amsthm} 

\newtheorem{theorem}{Theorem}[section]
\newtheorem{lemma}[theorem]{Lemma}

\newcolumntype{M}[1]{>{\centering\arraybackslash}m{#1}}
\newcolumntype{N}{@{}m{0pt}@{}} 

\newcolumntype{x}[1]{>{\centering\arraybackslash\hspace{0pt}}p{#1}}

\newif\ifhyper
\hypertrue
\ifhyper
\hypersetup{
   citecolor = {green},
   colorlinks = {true},
   urlcolor = {blue}
}
\fi

\usepackage{color}
\usepackage{soul}

\def\BState{\State\hskip-\ALG@thistlm}

\begin{document}

\title{Random access codes via  quantum contextual redundancy}

\author{Giancarlo Gatti}
\thanks{gatti.gianc@gmail.com}
\affiliation{Department of Physical Chemistry, University of the Basque Country UPV/EHU, Apartado 644, 48080 Bilbao, Spain}
\affiliation{EHU Quantum Center, University of the Basque Country UPV/EHU}
\affiliation{Quantum MADS, Uribitarte Kalea 6, 48001 Bilbao, Spain}
\orcid{0000-0002-5867-3537}
\author{Daniel Huerga}
\thanks{Present address: Stewart Blusson Quantum Matter Institute, University of British Columbia, 2355 E Mall, Vancouver BC, V6T 1Z4 Canada; huerga.daniel@gmail.com}
\affiliation{Department of Physical Chemistry, University of the Basque Country UPV/EHU, Apartado 644, 48080 Bilbao, Spain}
\orcid{0000-0003-3210-2370}
\author{Enrique Solano}
\thanks{enr.solano@gmail.com}
\affiliation{Department of Physical Chemistry, University of the Basque Country UPV/EHU, Apartado 644, 48080 Bilbao, Spain}
\affiliation{International Center of Quantum Artificial Intelligence for Science and Technology (QuArtist) \\ and Department of Physics, Shanghai University, 200444 Shanghai, China}
\affiliation{IKERBASQUE, Basque Foundation for Science, Plaza Euskadi 5, 48009 Bilbao, Spain}
\affiliation{Kipu Quantum, Greifswalderstrasse 226, 10405 Berlin, Germany}
\orcid{0000-0002-8602-1181}
\author{Mikel Sanz}
\thanks{mikel.sanz@ehu.eus}
\affiliation{Department of Physical Chemistry, University of the Basque Country UPV/EHU, Apartado 644, 48080 Bilbao, Spain}
\affiliation{EHU Quantum Center, University of the Basque Country UPV/EHU}
\affiliation{IKERBASQUE, Basque Foundation for Science, Plaza Euskadi 5, 48009 Bilbao, Spain}
\affiliation{Basque Center for Applied Mathematics (BCAM), Alameda de Mazarredo 14, 48009 Bilbao, Basque Country, Spain}
\orcid{0000-0003-1615-9035}

\date{(Dated: January 11, 2023)}

\begin{abstract}

We propose a protocol to encode classical bits in the measurement statistics of {many-body Pauli observables, leveraging quantum correlations for a random access code. Measurement contexts built with these observables yield outcomes with intrinsic redundancy, something we exploit by encoding the data into a set of convenient context eigenstates. This allows to randomly access the encoded data with few resources. The eigenstates used are highly entangled and can be} generated by a discretely-parametrized quantum circuit {of low depth}.
Applications of this protocol include algorithms {requiring large-data storage with} only partial retrieval, as is the case of decision trees.
Using $n$-qubit states, this {Quantum Random Access Code has greater success probability than its classical counterpart for $n\ge 14$ and than previous Quantum Random Access Codes for $n \ge 16$. Furthermore, for $n\ge 18$, it can be amplified into a nearly-lossless} compression protocol { with success probability $0.999$ and compression ratio $O(n^2/2^n)$. The data it can store is equal to Google-Drive server capacity for $n= 44$, and to a brute-force solution for chess (what to do on any board configuration) for $n= 100$.}
\end{abstract}

\maketitle

{
\section{Introduction}
}

Redundancy in classical and quantum information is generally used towards error-correction and data compression. It allows efficient schemes to protect~\cite{shannon1948mathematical, huffman2012fundamentals} or compress~\cite{al2008novel} classical data. In quantum systems, different  strategies based on classical redundancy~\cite{calderbank1996good, steane1996error} or non-local storage of information have been proposed for compressing~\cite{rozema2014quantum} and error-correcting~\cite{gottesman1996class, kitaev2003fault} quantum data.

{Redundancy can arise from certain types of quantum measurements. As noted by the Kochen-Specker theorem~\cite{peres2006quantum}, if we define a measurement context with a closed set of commuting observables, not all imaginable outcomes are possible. The outcomes are correlated, so an enumeration of them is redundant. This can be exploited for an encoding where lucky bitstrings are mapped into fewer (quantum) resources, but the same could be said of any $\{0,1\}^{a} \rightarrow \{0,1\}^{a+b}$ Boolean function using classical systems only, so applications of this are not immediately evident.}

{There is precedent for encodings between classical and quantum channels} in quantum teleportation \cite{bennett1993teleporting} and superdense coding \cite{bennett1992communication}, {where their relative entanglement-assisted capacity is at most in a 2:1 ratio ($2$ bits are used to teleport $1$ qubit, or $1$ qubit can encode $2$ bits)} \cite{bennett2002entanglement}. In contrast to these digital encodings, a statistical approach could also be used, with intrinsic error probability.

{A Random Access Code (RAC) is precisely that.} It is a communication protocol in which a bitstring of data is encoded into less information units, and then a fragment of the original data is randomly accessed, retrieving it back with some success probability. Quantum and Entanglement-assisted Random Access Codes \cite{wiesner1983conjugate,ambainis1999dense,ambainis2002dense,pawlowski2010entanglement} use quantum systems to store the data, and achieve a slight quantum advantage for small scales, but with a fast-decaying success probability as quantum system size increases \cite{casaccino2008extrema, tavakoli2015quantum}. {In these quantum versions, the measurement-basis choice is an integral part of the data retrieval.}

In this Article, we propose a new kind of Quantum Random Access Code (QRAC) {enhanced by measurement-context redundancy. This \textit{contextual QRAC} encodes classical bits in $n$-qubit quantum states, and accesses the bits by choosing a measurement context of $n$-body Pauli observables. The quantum states used are eigenstates of these observables, chosen to minimize the sampling requirement for retrieval. We show a basic implementation of this QRAC, storing $40$ binary digits of $\pi/4$ in $4$-qubit states.} We then provide a statistical analysis {of performance metrics} for arbitrary $n$. We prove that, for large $n$, our compression ratio scales with $O(n^2/2^n)$ {and the asymptotic success probability is greater than $1/2$. More specifically, for $n \gtrsim 10$}, we store $O(3^n)$ bits in $O(n (3/2)^n)$ $n$-qubit states, allowing random access to any portion of the data at a time with a success probability that does not depend on $n$. {This way, we showcase success probability higher than RACs for $n \ge 14$ and higher than existing QRACs for $n\ge 16$. For $n \ge 18$, our code can be amplified to near-perfect success probability ($\mathcal{F}=0.999$) while maintaining compression (encoding the data into less information units).} 

\vspace{0.3cm}
{In section \ref{sec:Preliminaries}~\textbf{Preliminaries} we present a basic QRAC scheme and the measurement bases used in previous literature. In section \ref{sec:ContextualQRAC}~\textbf{Contextual QRAC} we explain our proposal, a QRAC based on many-body Pauli observables, and show how it can exploit contextual redundancy. In section \ref{sec:Countingcontexts}~\textbf{Counting contexts} we determine the number and size of complete sets of commuting observables of $n$-body Pauli observables, which are required to compute the asymptotic performance of our QRAC. In section \ref{sec:Implementation}~\textbf{Implementation}, we construct a $4$-qubit implementation. In section \ref{sec:Statisticalextrapolation}~\textbf{Statistical extrapolation} we use statistical models to compute the sampling requirement and success probability for large number of qubits. In section \ref{sec:WrappinguptheQRAC}~\textbf{Wrapping up the QRAC} we summarize the asymptotic results of the previous section and compute the compression ratio of our QRAC. In section \ref{sec:Comparison}~\textbf{Comparison} we compare against previous RACs and QRACs, and determine advantage regimes. In section \ref{sec:Applications}~\textbf{Applications} we explain potential applications for large-scale Contextual QRACs. In section \ref{sec:Conclusionandoutlook}~\textbf{Conclusion and outlook} we summarize the results and address future research lines towards scalability.
}

\vspace{0.5cm}
{
\section{Preliminaries}
\label{sec:Preliminaries}

\subsection{Quantum Random Access Code}
Consider a scenario where Alice has $m$ bits of data that she wants Bob to randomly access. However, Alice can only send $k$ two-level systems to Bob, with $m > k$. Once Alice finds a smart protocol to do this, what is the probability of Bob to recover an arbitrary piece of the data? This situation depicts a $m\rightarrow k$ RAC. For instance, in a $3 \rightarrow 1$ RAC Alice has three bits which she wants Bob to access with maximum success probability, but she can only communicate a one-bit message. A strategy for this code is to send the most common bit of the three. This way, averaging over all $8$ possible scenarios (three bits), Bob has a $0.75$ probability of guessing any one of the bits, although his guess will be the same no matter which of the three bits he is trying to guess. 

A QRAC  \cite{wiesner1983conjugate,ambainis1999dense,ambainis2002dense,pawlowski2010entanglement} is a RAC where the $k$ two-level systems sent are quantum systems. We maintain the $m \rightarrow k$ notation, but notice that it is incomplete to fully specify a QRAC, as it does not define the maximum entanglement size used or the amount of data Bob can access at a time, which is limited due to the destructive nature of quantum measurements. In general, a $m \rightarrow k$ QRAC can outperform the success probability of a $m \rightarrow k$ RAC. For instance, in a $3 \rightarrow 1$ QRAC Alice has three bits and is allowed to send a single qubit to Bob, on which he can perform measurements. A way to do this is for Alice to send one out of $8$ possible quantum states, depending on what her $3$-bit data is. The eligible states can be defined embedding a cube in Bloch's sphere with sides facing the axes. Each corner of the cube points at one quantum state with a preference for $\ket{+}$ over $\ket{-}$ (or vice versa), $\ket{L}$ over $\ket{R}$ (or vice versa) and $\ket{0}$ over $\ket{1}$ (or vice versa). Here, $\ket{\pm}=\tfrac{1}{\sqrt{2}}(\ket{0}\pm\ket{1})$, $\ket{L}=\tfrac{1}{\sqrt{2}}(\ket{0}+i\ket{1})$ and  $\ket{R}=\tfrac{1}{\sqrt{2}}(\ket{0}-i\ket{1})$. Then, Bob can measure the quantum state with observables $X$, $Y$ or $Z$, to discover one of the three preferences with a success probability of $0.789$. This way, Alice can send one qubit and allow Bob to discover one of her three bits with that success probability, which is greater than that of the $3 \rightarrow 1$ RAC. Note that this time, Bob performs different measurements depending on which bit he is trying to retrieve.

As can be seen, QRACs showcase quantum advantage over their classical counterpart when employing systems as small as one-qubit. An immediate question is the scalability of this advantage for larger quantum systems.

Casaccino et. al. \cite{casaccino2008extrema} have studied $(d+1) \,\text{dits} \rightarrow 1 \,\text{qudit}$ QRACs where Alice and Bob use $d$-level systems instead of bits and qubits, i.e. $(d+1) \,\text{log}_2(d) \rightarrow \text{log}_2(d)$ QRACs in bit-qubit notation. They computed the dit-retrieval success probability up to $d=8$, reproducing the aforementioned $3 \rightarrow 1$ QRAC (for $d=2$) which is known to have advantage. However, they also showed that the success probability decays as $d$ increases, which apparently argues against the usefulness of high-dimensional QRACs (we further discuss this in the Comparison section). 

Similarly, Tavakoli et. al. \cite{tavakoli2015quantum} studied $2 \,\text{dits} \rightarrow 1 \,\text{qudit}$ and $3 \,\text{dits} \rightarrow 1 \,\text{qudit}$ QRACs and compared against the respective RACs, showing some quantum advantage but with diminishing returns: the ratio of success probabilities (quantum divided by classical) decreases for larger ($d\gtrsim 6$) system sizes.

Success probabilities of another kind of high-dimensional QRAC have been computed (bounded) by Pauwels et. al. \cite{pauwels2022almost}, namely the $3 \,\text{trits} \rightarrow 1 \, \text{qudit}$ scenario with $d$ up to $20$. In this case, success probability increases with $d$, as Alice has a fixed amount of data and a message size which depends on $d$. Reasonably, we expect perfect success probability at $d=27$, which can store $3$ trits.

\begin{figure*}[t]
\centering
\hspace*{-0.7cm}
\includegraphics[scale=0.65]{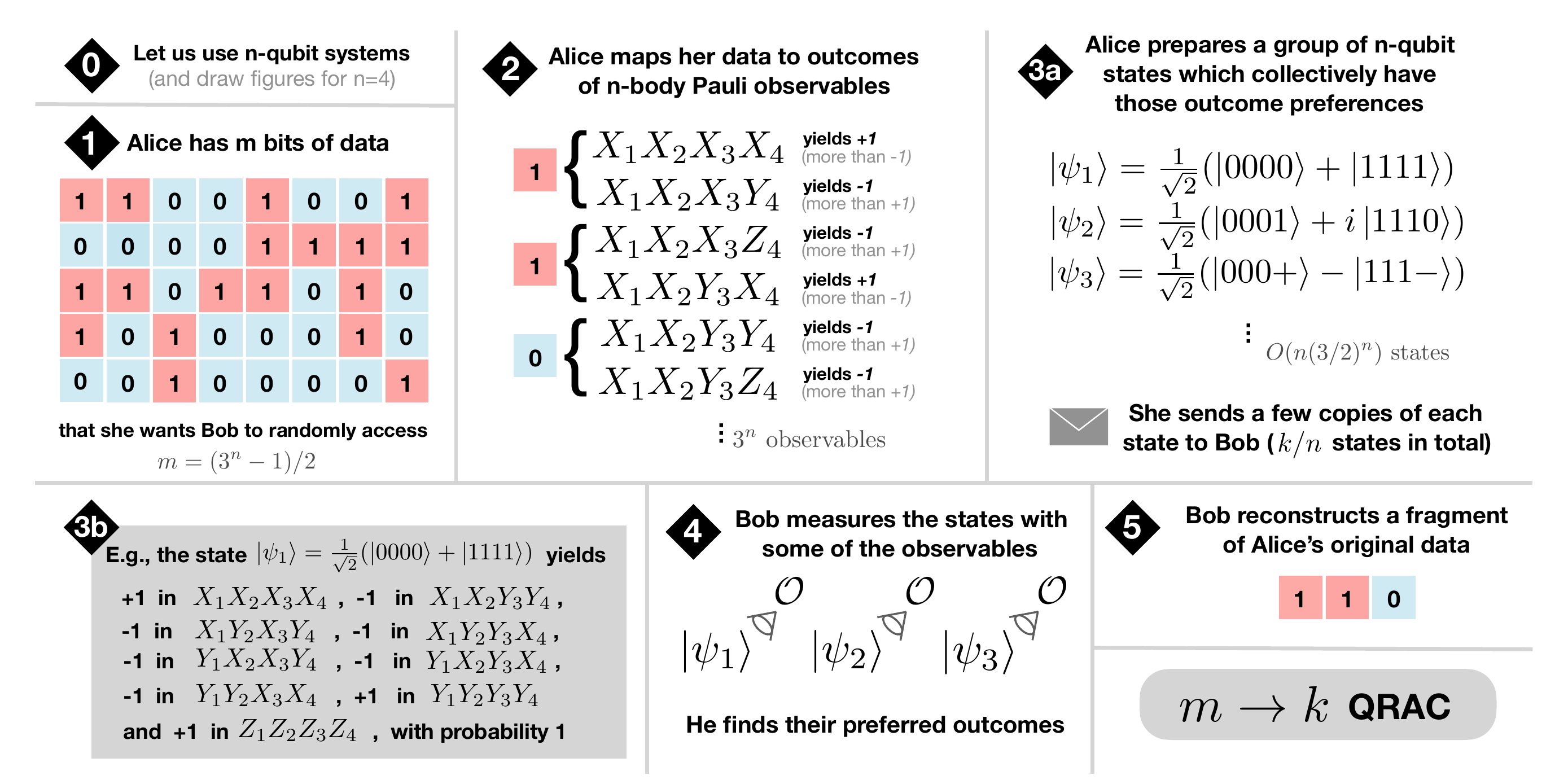}
\caption{Step-by-step summary of a $m\rightarrow k$ Quantum Random Access Code using $n$-qubit systems with $n$-body Pauli measurements. Alice encodes her data in the preferred outcomes of $n$-body Pauli observables, and sends to Bob a batch of $n$-qubit states which collectively yield that preference. In total, she sends $k/n$ states of $n$ qubits, counting all the different states sent and their copies. Bob retrieves a fragment of her data by measuring the states in some observables of his choosing.}
\label{protosummary}
\end{figure*}

\subsection{Mutually Unbiased Bases (MUB)}

In a QRAC, Bob chooses a basis to measure the quantum system sent by Alice. This choice depends on which of her bits (or dits) he wants to recover. So far, small and high-dimensional QRACs have considered Mutually Unbiased Bases (MUBs) as choices for Bob. We say that the eigenbases of observables $\mathcal{O}_A$ and $\mathcal{O}_B$ are mutually unbiased if all the eigenstates of $\mathcal{O}_A$ yield a homogeneous probability distribution when measured in $\mathcal{O}_B$. In other words, an eigenstate of $\mathcal{O}_A$ has no bias in the eigenbasis of $\mathcal{O}_B$ (and vice versa). The largest set of MUBs that can be formed grows linearly with the quantum system level $d$. Specifically, systems with prime-power level $d$ can be measured in $d+1$ MUBs~\cite{wootters1989optimal}, which includes any $n$-qubit system ($d=2^n$). 

Note that this determines the number of MUBs, but does not provide a way to obtain them. This result is what allows Casaccino et. al. \cite{casaccino2008extrema} to build QRACs where Alice has as many as $d+1$ dits and grants random access to one by sending a qudit to Bob, who picks one of the $d+1$ MUBs. However, the construction of these $(d+1) \,\text{dits} \rightarrow 1 \, \text{qudit}$ QRACs has proved to be a hard mathematical problem, analytically and computationally, and their asymptotic success probability is yet to be determined (we just know it for $d \le 8$). On high-dimensional $m \rightarrow k$ QRACs, the study of this regime of high compression ($k/m \ll 1$), where Alice has much more data than she sends to Bob, is still missing.

\section{Contextual QRAC}
\label{sec:ContextualQRAC}

In this Article, we propose the use of a different kind of observables, which appear in the field of contextuality: many-body Pauli observables. Unlike MUBs, which are unbiased and hard to compute, many-body Pauli observables share eigenstates between them but are a readily-defined set. Furthermore, unlike MUBs, each observable belongs to multiple measurement contexts. We employ these observables to build a $m \rightarrow k$ QRAC with states of arbitrary-$n$ qubits, which we present in this section.

\subsection{Overview}

We propose a special kind of $m\rightarrow k$ QRAC (Fig. \ref{protosummary}) based on many-body Pauli observables, where Alice has $m$ bits of data that she wants Bob to randomly access. To do this, she sends $N$ states of $n$ qubits, including duplicates that may be needed. The total number of qubits sent is $k=N n$, with maximum entanglement size $n$. We consider a fixed $n$ from which we determine $m$ and $k$, the latter also in terms of the code success probability. To randomly access Alice's data from these states, Bob measures with $n$-body Pauli observables}
\begin{equation}
\lbrace X,Y,Z\rbrace^{\otimes n}\text{.}
\label{paulis}
\end{equation}

These $3^n$ observables have two eigenvalues each, $\pm 1$, and thus are massively degenerate with dimension $2^{n-1}$. {We will refer to these outcomes as \textit{parities}, since canonical states measured in $Z^{\otimes n}$ yield $+1$($-1$) when they have an even(odd) number of $1$s (e.g. $\ket{1001}$ yields $+1$ and $\ket{1011}$ yields $-1$).}

This way, a set of states sent by Alice could communicate $3^n$ bits to Bob, defined by the preferred outcome of each observable. {For instance, in the $n=2$ case, if Alice sent a few copies of the states $\ket{\psi_1}=\ket{00}$, $\ket{\psi_2}=\ket{0L}$ and $\ket{\psi_3}=\tfrac{1}{\sqrt{2}}(\ket{+L}+\ket{-R})$, Bob could measure the states in one out of nine observables: $X_1 X_2$, $X_1 Y_2$, $X_1 Z_2$, $Y_1 X_2$, $Y_1 Y_2$, $Y_1 Z_2$, $Z_1 X_2$, $Z_1 Y_2$ or $Z_1 Z_2$. If he chose $Z_1 Z_2$, he would determine that $\ket{\psi_1}$ and $\ket{\psi_3}$ always output $+1$, i.e. they have a well-defined outcome, whereas $\ket{\psi_2}$ outputs both $+1$ and $-1$ with equal probability, i.e. it is unbiased. This way, Bob has learned that the preferred outcome of $Z_1 Z_2$ in this set of states is $+1$, equivalent to one bit of information. Other observables, such as $Z_1 X_2$, do not have a preferred outcome in this set of states.

We will now allow some freedom in our $m\rightarrow k$ QRAC. We choose Alice's data to be size $m=(3^n-1)/2$ by associating every bit to two Pauli observables instead of one, which leaves one of the $3^n$ observables out. This introduces a wiggle room in the encoding, to increase the success probability of our code. In the end, if this probability is high enough, our main metric of interest becomes the compression ratio $k/m$. 

Also note that the $m$ bits held by Alice are completely arbitrary, so we can assume this data to already be in its classical lossless-compression limit. Consequently, if the success probability of our QRAC aproaches $1$, any compression ratio $k/m<1$ that we achieve will be on top of the classical limit, i.e. a quantum advantage.}

{
\subsection{Contextuality}

We have proposed to use a group of $n$-qubit states $\{\ket{\psi_i}\}_i$ to store data in their collective $n$-body Pauli statistics, that is, in whether each of the $3^n$ observables yields preference for $+1$ or $-1$ outcomes (\textit{parities}). However, as studied in the field of quantum contextuality, there exist restrictions that outcomes of a complete set of commuting observables must adhere to, like the ones presented in Fig. \ref{MagicSquare}. As a rule of thumb, any $n$-qubit state should be able to yield only $2^n$ different outcomes. Thus, if a measurement context is defined by more than $n$ (commuting) observables with $1$-bit outputs, as is the case for $n$-body Pauli observables, they are necessarily correlated, and the information given by enumerating all the outputs of the context is redundant (\textit{contextual redundancy}). 

In this section, we show that all $2^{3^n}$ preferred parities (two options per observable) can be produced by the collective statistics of a set of $n$-qubit states $\{\ket{\psi_i}\}_i$. Afterwards, we establish a sampling requirement metric for mixed states, which proves useful despite our QRAC not using mixed states. Then, we introduce \textit{context eigenstates} as convenient states $\{\ket{\psi_i}\}_i$ for the encoding. Finally, we show in a $2$-qubit simulation that some preferred parities are more prevalent in the mixed-state phase space and are easier to sample from, particularly those satisfying more Magic Square restrictions (Fig. \ref{MagicSquare}).}

{
\subsubsection{Producing all preferred parities}
}

\begin{theorem}
{A set of $n$-qubit states $\{\ket{\psi_i}\}_i$ can have any collective (majority) bias in the outcomes of $n$-body Pauli observables.}
\end{theorem}

\begin{proof}
{The collective preferred parities of $N$ states $\{\ket{\psi_i}\}_i$ can be computed by the homogeneous mixed state $\rho=\tfrac{1}{N}\sum_i \ket{\psi_i} \bra{\psi_i}$, and $\rho$ can also be used to measure how strong the bias is. Trivially, we can show that in a sufficiently-large collection of states,} all preferred parities are possible, by considering a mixed state $\rho_\text{test}$ of $3^n$ states $\ket{\phi_j}$, each one eigenstate to a single $n$-body Pauli observable and unbiased for the rest. This allows independent control for the outcome preference of each observable. E.g., the state $\ket{++++}$ ($\ket{+++-}$) is eigenstate to $X_1 X_2 X_3 X_4$ with eigenvalue $+1$ ($-1$), and both states yield unbiased statistics on all other $4$-body Pauli observables. The associated mixed state $\rho_\text{test}=\tfrac{1}{3^n}\sum_j \ket{\phi_j} \bra{\phi_j}$ can yield any of the $2^{3^n}$ preferred parities.
\end{proof}

Note, however, that the example above yields a mixed state with probability distribution $\{1/2\pm 1/3^n,1/2 \mp 1/3^n\}$ (for $\{+1,-1\}$ outcomes, respectively) on all $n$-body Pauli observables. Intuitively, these distributions have a weak bias, such that determining the preferred parity requires many samples of the state. 

{
\subsubsection{Sampling requirement metric}

The retrieval protocol of our QRAC does not use mixed states, but instead allows Bob to measure a number of copies of each state $\ket{\psi_i}$ individually. However, a sampling requirement metric based on mixed states is convenient, as it can be used for comparative purposes. Given a set of $N$ states $\{\ket{\psi_i}\}_i$, the homogeneous mixed state $\rho=\tfrac{1}{N}\sum_i \ket{\psi_i} \bra{\psi_i}$ has a probability distribution $\{p_\mathcal{O}^+,p_\mathcal{O}^-\}$ for each $n$-body Pauli observable. For each of them, we can compute the required number of samples to guarantee with success probability $f$ that the majority yield the preferred parity of $\mathcal{O}$,
\begin{equation}
S_\text{mix}(p_\mathcal{O},f)= \underset{s\in \mathbb{N}}{\text{arg min}} \left(\sum_{j=\left \lfloor{s/2}\right \rfloor +1}^s B(s,p_\mathcal{O};j) >f\right),
\label{eq:SR_mixed}
\end{equation}
where $p_\mathcal{O}=\max(p_\mathcal{O}^+,p_\mathcal{O}^-)$ and $B(s,p_\mathcal{O};j)={s \choose j } p_\mathcal{O}^j (1-p_\mathcal{O})^{s-j}$ is the binomial density function for $j$ successes in $s$ shots with probability $p_\mathcal{O}$. Then, this sampling requirement can be averaged over all $3^n$ observables, obtaining $\langle S_\text{mix}\rangle$, which is a general sampling requirement metric for the homogeneous mixed state of $\{\ket{\psi_i}\}_i$ in all $n$-body Pauli observables.}

{
\subsubsection{Context eigenstates}

We have proved that all preferred parities can be produced, but a more efficient construction with lower sampling requirement is still required. For an efficient encoding of our QRAC, instead of using states which are eigenstate to a single observable and unbiased in the rest, we use states which are eigenstates to as many observables as possible (and unbiased in the rest). To do this, we must consider sets of commuting $n$-body Pauli observables. In general, commuting operators have common eigenbases. This eigenbasis is unique in the case of complete sets of commuting $n$-body Pauli observables. We call these commuting sets \textit{contexts} --short for \textit{measurement contexts}--, and refer to the common set of eigenstates as context eigenstates. Within the limits of $n$-body Pauli observables, context eigenstates are unbiased outside of their context.

Context eigenstates of two qubits have eigenvalues restricted by the same conditions of the Mermin-Peres Magic Square. For instance, the Bell state $\tfrac{1}{\sqrt{2}}(\ket{00}+\ket{11})$ is a context eigenstate to the context $\{X_1 X_2, Y_1 Y_2, Z_1 Z_2\}$, with respective eigenvalues $\{+1,-1,+1\}$. The rest of the Bell basis shares the same context, with eigenvalues that also multiply to $-1$. On the other hand, the $i$-phase Bell state $\tfrac{1}{\sqrt{2}}(\ket{00}+i\ket{11})$ is a context eigenstate to the context $\{X_1 Y_2, Y_1 X_2, Z_1 Z_2\}$, with respective eigenvalues $\{+1,+1,+1\}$. The rest of the $i$-phase Bell basis shares the same context, with eigenvalues that also multiply to $+1$. Note that --in contrast to MUBs-- an observable can belong to multiple contexts, such that Alice can express the parities she wants with the most convenient context eigenstates.}



\begin{figure}[t]
\centering
\hspace*{-0.31cm}
\includegraphics[scale=0.65]{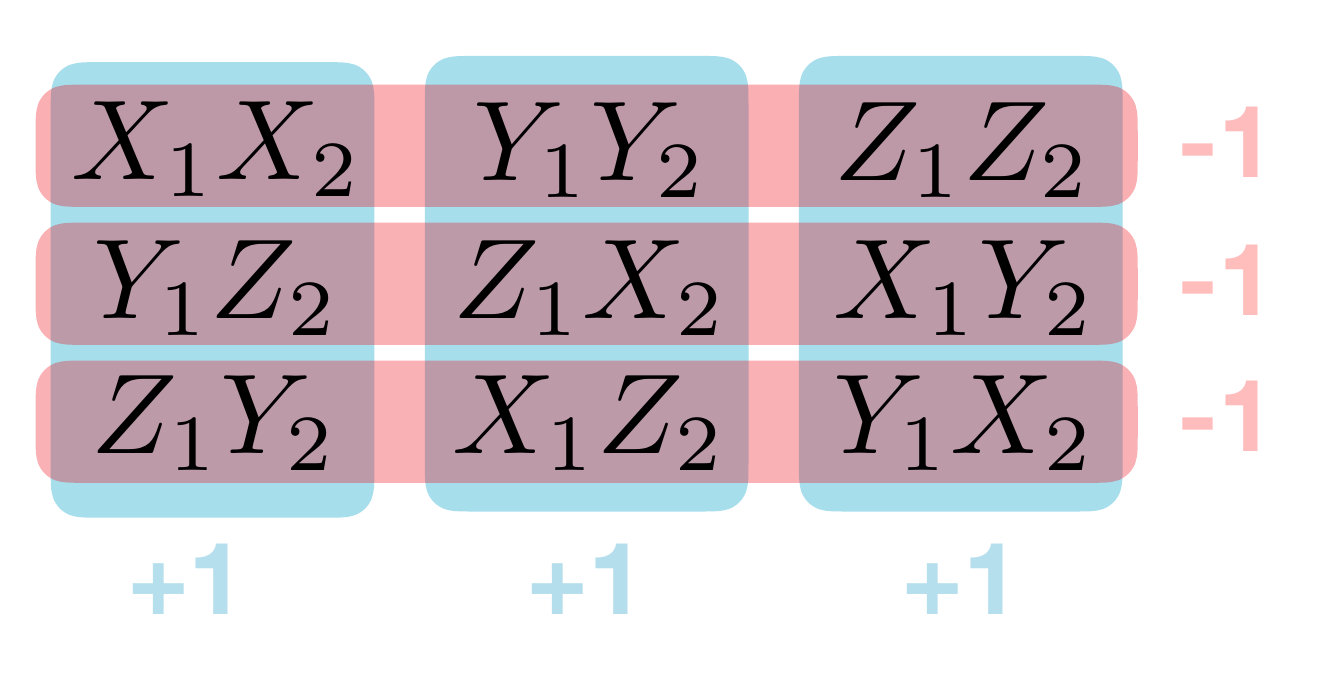}
\caption{A version of the Mermin-Peres Magic Square~\cite{peres2006quantum} using two-body Pauli observables. This grid of observables is organized in such a manner that any row or column is a complete set of commuting observables, a \textit{measurement context}. Any quantum state measured in one of these rows(columns) will collapse onto a context eigenstate and yield outcomes which multiply to $-1$($+1$). In other words, not all imaginable outcomes of a measurement context are possible; there are restrictions.}
\label{MagicSquare}
\end{figure}

{
\subsubsection{Contextual redundancy}

Consider a configuration of preferred parities on all $3^n$ observables, that is, $+1$ or $-1$ biases that we want on each $n$-body Pauli observable. For $n=2$, due to the Mermin-Peres Magic Square, one would expect that preferred parities that satisfy more measurement context restrictions (the row or column outcome product) should be possible to express with context eigenstates instead of eigenstates to individual observables. This would produce probability distributions with more extreme biases, reducing the sampling requirement for retrieval. We test this conjecture by producing random $2$-qubit mixed states and classifying them according to their preferred parities, which have $2^{3^2}=512$ configurations. In the horizontal axis of Fig. \ref{samplesversusoccurrences}, we count how often each preferred-parity configuration occurs. In the vertical axis of Fig. \ref{samplesversusoccurrences}, we compute the mixed-state sampling requirement metric $\langle S_\text{mix}\rangle$ and show the smallest value reached by each preferred-parity configuration. Notably, the preferred parities occur more often and have lower sampling requirements when they satisfy more Magic Square restrictions, particularly if all rows or all columns are satisfied.

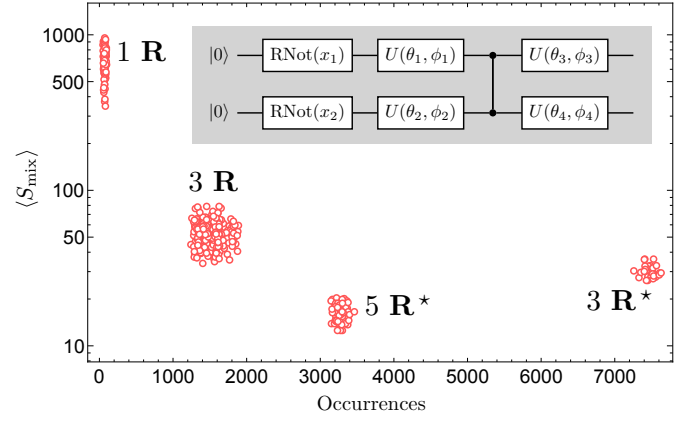
\begin{figure}[t]
\centering
\hspace*{-0.50cm}
\includegraphics[scale=0.38]{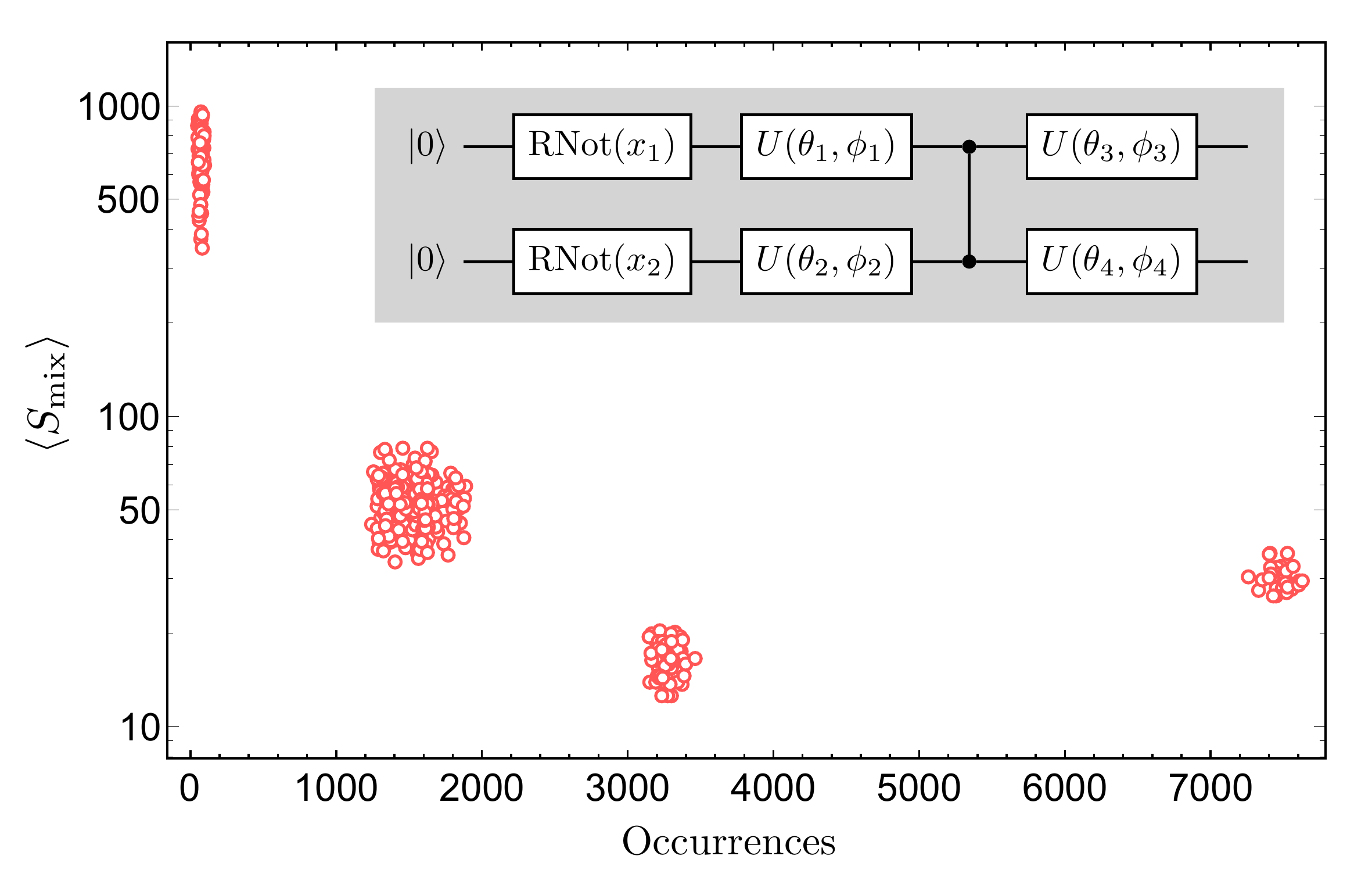}
\put(-217,138){\textbf{$1$} \textbf{R}}
\put(-190,89){\textbf{$3$} \textbf{R}}
\put(-124,42){\textbf{$5$} \textbf{R}$^{\,\star}$}
\put(-41,44){\textbf{$3$} \textbf{R}$^{\,\star}$}
\caption{{Using a parametrized circuit (shown in inset), $10^6$ random $2$-qubit mixed states were generated and classified according to their $2$-body-Pauli preferred outcomes ($512$ possible configurations), counting the number of occurrences of each configuration and taking note of the lowest average sampling requirement $\langle S_\text{mix}\rangle$ ($f=0.95$) of each configuration. We plot these two values, obtaining a graph with $512$ points.} In the quantum circuit, $U(\theta,\phi)=R_z (\phi) R_x (-\pi/2) R_z (\theta) R_x (\pi/2)$, with $\theta \in [0,\pi]$, $\phi \in [0,2 \pi]$. Classical randomness to generate mixed states is introduced with a \textit{random-not} gate: $\text{RNot}(x)=X$ with probability $x$, and RNot$(x)=\mathbb{1}$ otherwise. {The preferred parities are labeled according to the total number of Magic Square restrictions satisfied ($1 \text{ \textbf{R}}$, $3 \text{ \textbf{R}}$ or $5 \text{ \textbf{R}}$) and marked with a star when all rows or all columns are satisfied. Notably, the preferred parities cluster according to these labels.}}
\label{samplesversusoccurrences}
\end{figure}

Although this test has been performed for $n=2$, a general conclusion can be drawn: not all $2^{3^n}$ preferred parities are equally prevalent on the phase space of $n$-qubit mixed states, and some preferred parities allow for probability distributions which are easier to sample from. Thus, in the following section, we employ the wiggle-room granted by our $2$-to-$1$ encoding to map our data into the most convenient preferred parities. Specifically, preferred parities that can be expressed with context eigenstates.
}

{
\subsection{Encoding}
}

\begin{figure}
\hspace*{-0.40cm}
\begin{adjustbox}{width=0.50\textwidth}
\begin{quantikz}\lstick{$\ket{0}$} & \gate{H} & \ctrl{1} & \qw & \gate{S^\star(\alpha_0)} & \gate{Z^\star(\alpha_1)} & \gate{\Gamma(\beta_1)} & \qw \\
\lstick{$\ket{0}$} & \qw & \targ{} & \ctrl{1} & \qw & \gate{X^\star(\alpha_2)} & \gate{\Gamma(\beta_2)} & \qw \\
\lstick{$\ket{0}$}&\qw &\qw &\targ{} &\ctrl{1} & \gate{X^\star(\alpha_3)} & \gate{\Gamma(\beta_3)} & \qw\\
\lstick{$\ket{0}$}&\qw &\qw & \qw &\targ{} & \gate{X^\star(\alpha_4)} & \gate{\Gamma(\beta_4)} & \qw
\end{quantikz}
\end{adjustbox}
\caption{{Discretely-parametrized quantum circuit to generate eigenstates of contexts of maximum size for $n=4$. Depending on its parameters, this circuit constructs a GHZ or $i$-phase GHZ state and applies $X$ and $Z$ gates, as well as Pauli-basis transformations, allowing us to produce any context eigenstate.} Here, $H$ is the Hadamard gate, and $\alpha_i=\{0,1\}$ and $\beta_i=\{0,1,2\}$ are discrete parameters. The gates are defined such that $S^\star (0)=Z^\star (0)=X^\star (0)=\Gamma(0)=\mathbb{1}$ and $S^\star (1)=S=\ket{0}\bra{0} + i \ket{1}\bra{1}$, $Z^\star (1)=Z$, $X^\star (1)=X$, $\Gamma(1)=H.Z.S$ and $\Gamma(2)=S.H$. {This circuit layout can be easily extrapolated to arbitrary $n$, showcasing  $2\times 6^n$ parameter configurations ($n+1$ $\alpha$-parameters and $n$ $\beta$-parameters), which matches the count of $2\times 3^n$ contexts with $2^n$ eigenstates each. For even $n>2$, all parameter configurations can be shown to produce different quantum states. Finally, note that only $n-1$ gates of $2$-qubits are used, so the depth of the circuit is at worst linear.}}
\label{maincircuit}
\end{figure}
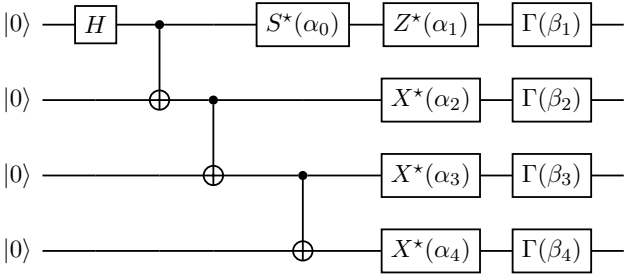

{In our code, Alice has $m=(3^n-1)/2$ bits of data $\bar{b}=\left\{ b_\ell \right\}_\ell$ that she encodes into preferred parities of $n$-body Pauli observables.} This is a $2$-to-$1$ mapping between data and observables, which allows her to have some wiggle room in choosing preferred parities. To specify which data bit corresponds to which observable pair, an ordered set of observable-couples $\mathcal{G}=\{\{\mathcal{O}_1,\mathcal{O}_2\},\{\mathcal{O}_3,\mathcal{O}_4\},...\}$ is agreed beforehand between Alice and Bob. {The order itself is irrelevant, but it is convenient for retrieval purposes that the couples commute.} This way, the $\ell$-th data bit $b_\ell$ is stored in the $\ell$-th couple of observables $\mathcal{G}_\ell$, leaving one observable out. Without loss of generality, we identify bit $b_\ell=0~(1)$ with equal ($=$) and different ($\ne$) preferred parities, respectively, when measuring the corresponding pair of observables. Each bit of data can be encoded in two ways, so Alice can choose from $2^m$ compatible preferred parities to represent her data, which is a super-exponential degree of freedom in $n$. To reduce the retrieval sampling requirement, Alice considers the preferred parities compatible with her data, and chooses the one with most redundant contexts, i.e. statistics resembling those of context eigenstates. {In the Implementation section, we do this for the $n=4$ case by constructing and maximizing a scoring function based on Hamming distances.} 

{Finally, Alice computes a selection of context eigenstates to send to Bob with collective statistics resembling her target preferred parities. These states can be generated and measured with discretely parametrized quantum circuits of low depth, as shown in Fig.~\ref{maincircuit} and Fig. \ref{measurecircuit}. By computing the measurement statistics in terms of the parameters, states are chosen such that the selection size and retrieval sampling requirement are minimized}. An encoding example is given in Fig.~\ref{didacticencoding}.

\vspace*{0.5cm}
\subsection{Retrieval}

\begin{figure}
\centering
\hspace*{-0.50cm}
\begin{adjustbox}{width=0.39\textwidth}
\begin{quantikz}\lstick[4]{$\ket{\psi}$} & \gate{\Gamma(\beta_1)} & \ctrl{4} & \qw & \qw & \qw & \qw \\
\qw & \gate{\Gamma(\beta_2)} & \qw  & \ctrl{3} & \qw & \qw & \qw \\
\qw & \gate{\Gamma(\beta_3)} &\qw &\qw &\ctrl{2} & \qw & \qw\\
\qw & \gate{\Gamma(\beta_4)} &\qw & \qw &\qw  &\ctrl{1} & \qw\\
\lstick{$\ket{0}$} & \qw & \targ{} & \targ{} & \targ{} & \targ{} & \meter{}
\end{quantikz}
\end{adjustbox}
\caption{{Discretely-parametrized quantum circuit to measure an $n$-qubit state $\ket{\psi}$ with one $n$-body Pauli observable using one ancilla qubit, illustrated for $n=4$. Here, $\beta_i=\{0,1,2\}$ are discrete parameters and the gates are defined as $\Gamma(0)=\mathbb{1}$, $\Gamma(1)=H.Z.S$ and $\Gamma(2)=S.H$, such that they perform transformations between Pauli bases. The circuit can be easily extrapolated to arbitrary $n$, showcasing $3^n$ parameter configurations ($n$ $\beta$-parameters), i.e. one per observable.}}
\label{measurecircuit}
\end{figure}
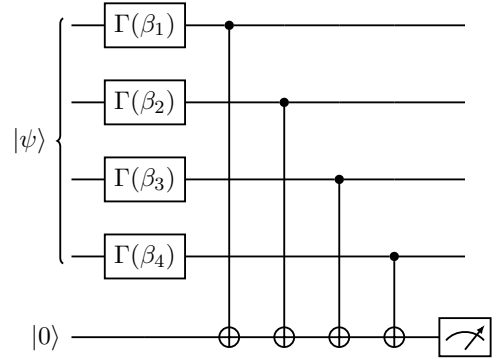

In this protocol, Bob sequentially measures the selection of states sent by Alice in an observable $\mathcal{O}$ of his choosing, filtering out those with unbiased statistics. {This measurement is performed with the circuit in Fig.~\ref{measurecircuit}. Recall that Alice sends context eigenstates, and thus they yield outcomes with either well-defined or fully unbiased probability distributions. In $T$ measurement steps, Bob measures at most that many copies of each state. We consider a sufficiently large $T$ to} statistically distinguish well-defined from unbiased parities. {If the prepared states are ideal, they are filtered out if they yield a changing outcome at any point, since Bob can immediately conclude that those states are not well-defined, as shown in Fig.~\ref{didactretrieval}. In that case, unbiased states are filtered out exponentially fast. On the other hand, if the prepared states are noisy, a threshold to distinguish nearly-well-defined from almost-unbiased states must be considered, and the number of required samples would increase. Both for ideal or noisy states, once this filtering is performed,} the most frequent outcome among unfiltered states defines the preferred parity retrieved by Bob.

Repeating the measuring process for part(all) of the observables, the original data can be partially(completely) retrieved. Note that commuting observables can be measured using the same samples.

\begin{figure}[t]
\hspace*{-0.30cm}
\includegraphics[width={0.50 \textwidth}]{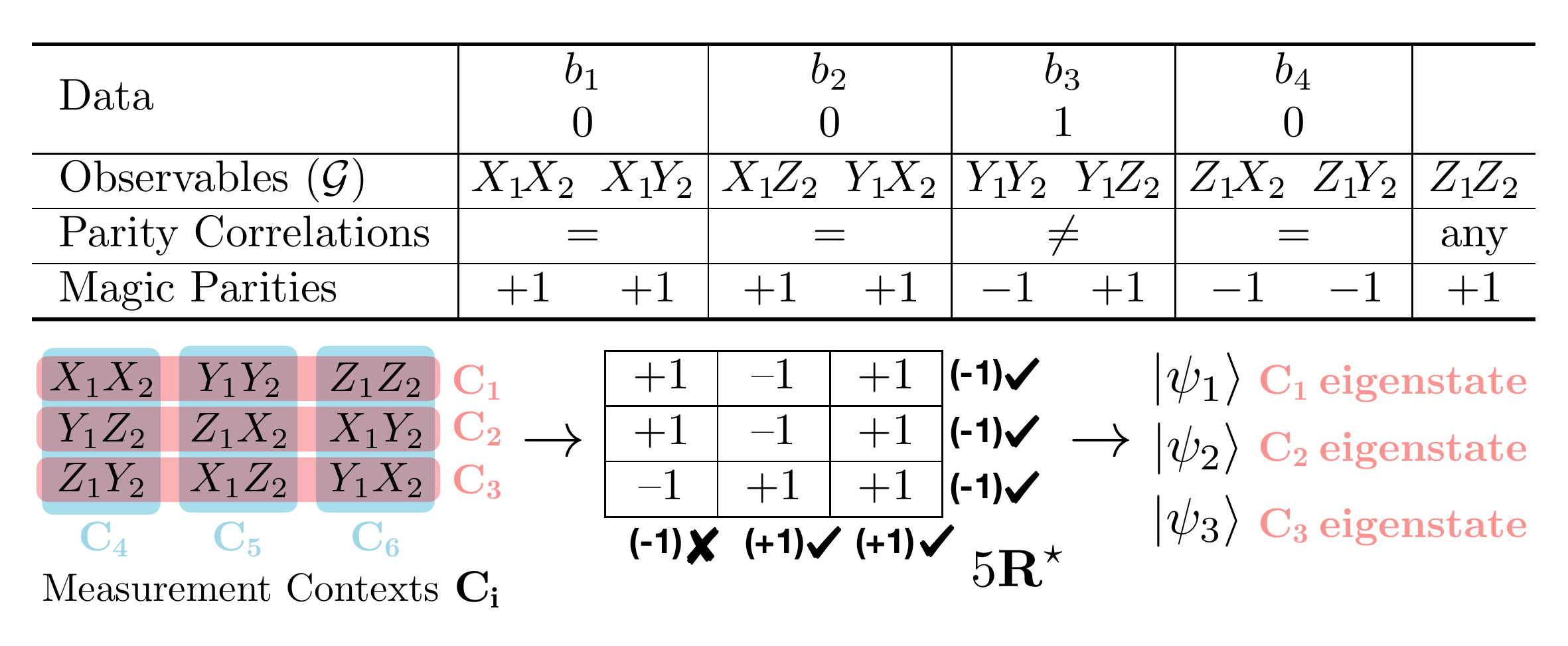}
\caption{Example of a $4$-bitstring encoded by three two-qubit context eigenstates: $\ket{\psi_1}=\tfrac{1}{\sqrt{2}}(\ket{00}+\ket{11})$, $\ket{\psi_2}=\tfrac{1}{\sqrt{2}}(\ket{L0}-\ket{R1})$
and 
$\ket{\psi_3}=\tfrac{1}{\sqrt{2}}(\ket{0R}+\ket{1L})$, 
where 
$\ket{\pm}=\tfrac{1}{\sqrt{2}}(\ket{0}\pm\ket{1})$, $\ket{R}=\tfrac{1}{\sqrt{2}}(\ket{0}- i \ket{1})$ and $\ket{L}=\tfrac{1}{\sqrt{2}}(\ket{0}+i \ket{1})$.
Let Alice's data be $\bar{b}=0010$ and let the order of observable-couples $\mathcal{G}$ agreed beforehand be the alphabetical order. This means that $X_1 X_2$ and $X_1 Y_2$ must have equal preferred parities (bit 0), $X_1 Z_2$ and $Y_1 X_2$ must have equal preferred parities (bit 0), $Y_1 Y_2$ and $Y_1 Z_2$ must have different preferred parities (bit 1), $Z_1 X_2$ and $Z_1 Y_2$ must have equal preferred parities (bit 0), and $Z_1 Z_2$ may have any parity. There are several  preferred parities compatible with the data, but Alice chooses one which satisfies five Magic-Square restrictions, allowing her to express the $9$ preferred parities with three row context eigenstates.}
\label{didacticencoding}
\end{figure}

{In the Statistical Extrapolation section, we consider ideal states to model the required number of samples for any number of qubits $n$. In this scenario, the only error source in the retrieval are the remaining unfiltered unbiased states.}

{
\section{Counting contexts}
\label{sec:Countingcontexts}

To analyze the performance of our QRAC for arbitrary number of qubits $n$, we need to count the number of available contexts and context eigenstates. For simplicity, we only consider even $n$, as odd-sized Pauli observables introduce additional considerations.

Using $n$-body Pauli observables, the number of contexts and their sizes are trivial for $n=2$. In the Mermin-Peres Magic Square, there are $6$ contexts (complete sets of commuting observables), and each one has $3$ observables. For $n=4$, there are $270$ contexts of size $9$ (which we prove later in this section). However, for $n>4$, contexts no longer have a unique size.}

{Large context sizes define eigenstates with well-defined parities in several observables simultaneously, so for arbitrary $n$ we focus on counting and constructing the largest contexts possible. This way, we prove, count and construct contexts of size $2^{n-1}+1$, which are of maximum size to the best of our knowledge when contrasted with numerical simulations of up to $n=20$.

\begin{figure}[t]
\centering
\includegraphics[width={0.46 \textwidth}]{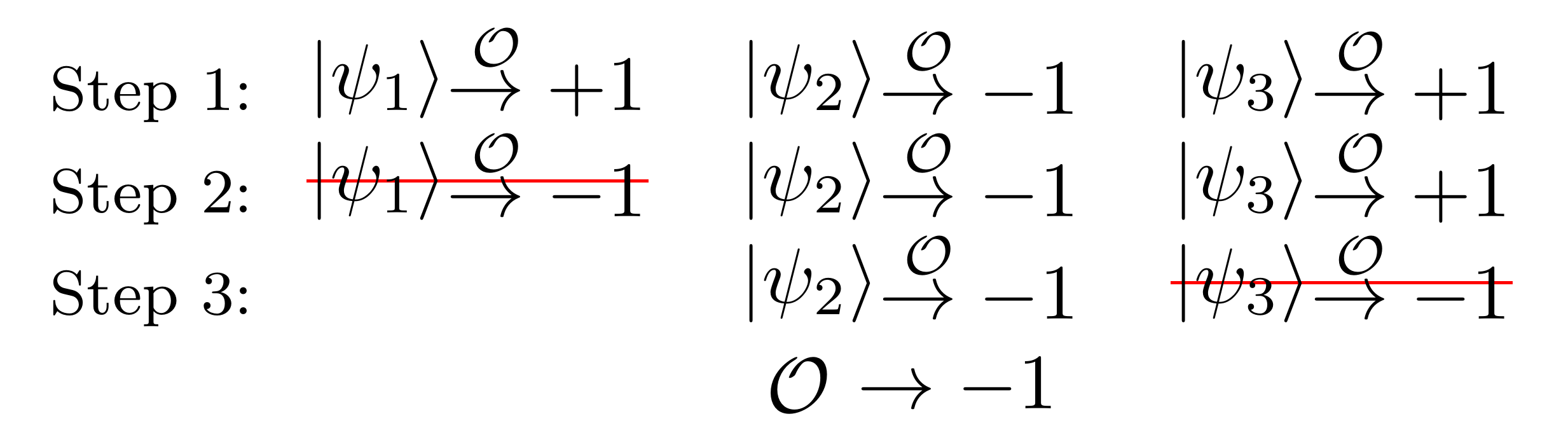}
\caption{Retrieval protocol for the preferred parity of observable $\mathcal{O}=Z_1 X_2$ from states $\ket{\psi_1}$, $\ket{\psi_2}$ and $\ket{\psi_3}$ of Fig.~\ref{didacticencoding}, with three measurements steps (requiring at most three copies of each state). The protocol consists in filtering out eigenstates with unbiased statistics in $\mathcal{O}$ by performing successive measurements and discarding those that change their outcome. Both states $\ket{\psi_1}$ and $\ket{\psi_3}$ have unbiased probability distributions when measured in $Z_1 X_2$, which is why they are likely to eventually change their outcome, whereas $\ket{\psi_2}$ always yields $-1$. In general, the expected number of unbiased states that fail to be filtered decreases exponentially in the number of retrieval steps.}
\label{didactretrieval}
\end{figure}

\begin{lemma}
For $n$-body Pauli observables with even $n$, each observable can be associated to two contexts of size $2^{n-1}+1$.
\label{thelemma1}
\end{lemma}

\begin{proof}
For even $n$, choose an arbitrary $n$-body Pauli observable $\mathcal{O}_0$, which we will refer to as our \textit{context generator}.
Without loss of generality, consider $\mathcal{O}_0=Z^{\otimes n}$. 
There are $2^n$ $n$-body Pauli observables which share no term with the generator in the same qubit positions,
\begin{equation}
\mathcal{C}'=\lbrace X,Y\rbrace^{\otimes n}.
\end{equation} 
Note that they all commute with $\mathcal{O}_0$ since $n$ is even.
Not all observables in $\mathcal{C}'$ mutually commute, but $\mathcal{C}'$ can be split into two subsets $\mathcal{C}'_1$ and $\mathcal{C}'_2$ of equal size where this is true. In our example, $\mathcal{C}'_1$($\mathcal{C}'_2$) comprises all the observables in $\mathcal{C}'$ with even(odd) number of Pauli $X$. In this manner, two contexts are defined, $\mathcal{C}_1=\{\mathcal{O}_0\} \cup \mathcal{C}'_1$ and $\mathcal{C}_2=\{\mathcal{O}_0\} \cup \mathcal{C}'_2$, of size $2^{n-1}+1$.
\end{proof}

With this method, each generator defines two contexts of size $2^{n-1}+1$, and $3^n$ observables can be used as generators. This would seem to prove the existence of $2 \times 3^n$ contexts, but we need to guarantee that none of them is counted twice. For $n \ge 4$, we prove in the following that all contexts have a unique generator.

\begin{lemma}
For $n$-body Pauli observables with even $n$, contexts of size greater than three can have at most one observable which shares no terms with the rest.
\label{thelemma2}
\end{lemma}

\begin{proof}
Consider a set of three $n$-body Pauli observables $\{\mathcal{O}_1,\mathcal{O}_2,\mathcal{O}_3\}$ sharing no terms with each other (in the same qubit positions). This is possible, as each qubit position can be $X$, $Y$ or $Z$. For even $n$, this set commutes, as its observables have even number of changes to each other, and $X$, $Y$ and $Z$ anticommute. If we were to add a fourth different $n$-body Pauli observable $\mathcal{O}_4$ which commuted with the rest (even number of changes), it would necessarily share terms with at least two other observables in the set. This way, constructively, any context with size larger than three can have at most one observable sharing no terms with the rest.
\end{proof}

This implies that for $n \ge 4$ the \textit{context generators} in Lemma  \ref{thelemma1} produce unique contexts. Lemma \ref{thelemma1} and Lemma \ref{thelemma2} prove that:

\begin{theorem}
For $n$-body Pauli observables with even $n \ge 4$, there exist at least $2\times 3^n$ contexts of size $2^{n-1}+1$.
\end{theorem}

In the case of $n=2$, different generators can produce the same contexts, so there are only $6$ contexts instead of $2\times 3^n$. In the case of $n=4$, the method in Lemma~\ref{thelemma1} constructs $2\times 3^n=162$ contexts of size $2^{n-1}+1=9$, which are the ones used in the implementation section, but partitioning the system into two $2$-qubit subsystems (which can be done in $3$ ways) and using the known $n=2$ contexts yields additional contexts of size $9$, namely $3\times 6 \times 6$ of them, to a total of $270$. Fortunately, this double source for maximum-sized contexts happens exclusively in $n=4$. For $n\ge 6$ there are exactly $2\times 3^n$ contexts of size $2^{n-1}+1$.

Each context has an associated basis of $2^n$ unique context eigenstates. For generator $\mathcal{O}_0=Z^{\otimes n}$, one of the context eigenstates of $\mathcal{C}_1$ is the $n$-qubit GHZ state $\tfrac{1}{\sqrt{2}}(\ket{00\ldots 0}+\ket{11\ldots 1})$, and one of the context eigenstates of $\mathcal{C}_2$ is the $n$-qubit $i$-phase GHZ state $\tfrac{1}{\sqrt{2}}(\ket{00...0}+i \ket{11...1})$. The rest of their bases can be obtained with local Pauli gates $X$ and $Z$. For any other generator, the context eigenstates of generator $\mathcal{O}_0=Z^{\otimes n}$ can be rotated into the corresponding basis. For instance, $\tfrac{1}{\sqrt{2}}(\ket{LLLL}+\ket{RRRR})$ and $\tfrac{1}{\sqrt{2}}(\ket{LLLL}+i \ket{RRRR})$ are eigenstates of contexts generated by $Y_1 Y_2 Y_3 Y_4$. This allows us to build a discrete-parameter circuit which can generate any context eigenstate (Fig. \ref{maincircuit}).}

{

\section{Implementation}
\label{sec:Implementation}

In the following, we construct a toy implementation of our QRAC for $n=4$ qubits. Alice grants random access to a data of $m=(3^n-1)/2=40$ bits by sending a selection of $4$-qubit states to Bob. We choose this data to be binary digits of $\pi/4$. The implementation consists in finding states from the quantum circuit in Fig. \ref{maincircuit} with statistics that collectively match the data. However, recall that the mapping between data and preferred parities is not straightforward. Every bit is associated to a pair of (four-body) Pauli measurements over the selection of states; if the pair yields equal preferred parities, the encoded bit is 0, and vice versa. Thus, we must first decide a configuration of preferred parities that we will try to target with the selection of states. This way, we have a two-step optimization: 1) choosing target preferred parities and 2) building up a selection of states to encode it. Then, we analyze the probability distributions of the resulting state selection and compute the retrieval sampling requirement.

\subsection{Target preferred parities}

To find convenient preferred parities for the encoding, we must perform an optimization on the phase space of compatible parities to our data. We use alphabetical order for forming observable couples $\mathcal{G}=\{\{\{X_1 X_2 X_3 X_4,X_1 X_2 X_3 Y_4\},...\}$. Each of the $40$ couples can encode a bit in one out of two ways, so there are $2^{40} \approx 10^{12}$ compatible preferred parities to our data. In other words, one compatible configuration can be expressed with $40$ bits (too). We do not explore them all, but instead stochastically optimize this $40$-bitstring by doing random bit flips and maximizing a scoring function. 

We build the scoring function based on the Hamming distance (the number of differences) between eigenstate well-defined outcomes, i.e. their eigenvalues, and compatible preferred parities. For this, we first compute the eigenvalues of all context eigenstates. For $n=4$, we use $2\times 3^n=162$ contexts (of size $9$), and each one has $2^n=16$ eigenstates, so there are $2592$ context eigenstates in total. Computing all of their eigenvalues is not efficient for general $n$, but can be done for $n=4$. For our scoring function, we count the number of eigenstates with Hamming distance $i$ ($0 \le i \le 9$, as contexts are size $9$) and denote it by $\mathcal{M}_i$. The scoring function is computed as 
\begin{equation}
\mathcal{L}=\sum_i w_i \mathcal{M}_i \text{,}
\end{equation}
using the weights $w_0=100$, $w_1=10$, $w_2=1$, and $w_i=0$ for $i \ge 3$. Consequently, the scoring function $\mathcal{L}$ primarily rewards compatible configurations having several context eigenstates with perfect match, and then rewards eigenstates with lower number of matches. The result of this optimization is a configuration of $81$ target preferred parities, although arguably one observable is uncoupled, so its parity does not really matter.

\subsection{Eigenstate selection}

We need to build a selection of context eigenstates $\{\ket{\psi_i}\}_i$ which matches the target preferred parity that we obtained in the previous optimization. We denote the number of states in this selection by $N_\text{s}$. For $n=4$, there are $2592$ context eigenstates, so an arbitrary selection of them has a very large number of configurations. We can reduce this phase space by considering the best states from each context (the ones with most matches to the target preferred parity), to an eligible pool of $O(162)$ states, which is more manageable. Afterwards, we minimize the size of the selection $N_\text{s}$ and a sampling requirement metric based on the corresponding mixed state. As noted before, the retrieval protocol of our QRAC does not use mixed states, i.e. classical random ensembles of pure states, but rather Bob receives states which he can access individually. However, $\rho=\tfrac{1}{N_\text{s}}\sum_i \ket{\psi_i} \bra{\psi_i}$ can still serve to compute a sampling requirement metric, like the one in Eq. \eqref{eq:SR_mixed}. We use that metric, and for computation efficiency take the normal approximation, obtaining

\begin{equation}
S_\text{mix}(p_\mathcal{O},f)\approx\left[2 \left( \text{Erf} ^{-1} (2 f - 1)\right)^2\right] \frac{p_\mathcal{O}(1-p_\mathcal{O})}{(p_\mathcal{O}-1/2)^2}\text{,}
\label{SRnormal}
\end{equation}
where $\text{Erf} ^{-1}(\cdot)$ is the inverse of the error function, and $p_\mathcal{O}$ is the higher value of the probability distribution of observable $\mathcal{O}$ in $\rho$. Since we only need a metric for the overall probability-distribution values, we omit the term in brackets, and then use as metric the average for all observables $\mathcal{O}$. Furthermore, to avoid diverging values when $p_\mathcal{O}$ approaches $1/2$, we take a large number $u$ and use it as boundary. We also assign $u$ to our metric when the higher probability is not on the target parity. This way, for probability $p_\mathcal{O}$ of the target parity in observable $\mathcal{O}$, we have
\begin{equation}
        S_\text{mix}'=
        \left\{ \begin{array}{ll}
            \text{min}\left(u\text{, }\frac{ p_\mathcal{O} (1-p_\mathcal{O})}{(p_\mathcal{O}-1/2)^2} \right) & p_\mathcal{O}>1/2 \\
            u & p_\mathcal{O} \le 1/2
        \end{array} \right.
        \label{SRmetric} \text{,}
\end{equation}
and our sampling requirement metric is its average over all observables, $\langle S_\text{mix}'\rangle$.

Now we obtain the state selection cost function by multiplying the two quantities: 
\begin{equation}
        \mathscr{L}= N_\text{s} \langle S_\text{mix}'\rangle \text{.}
\end{equation}
This cost function allows to obtain a small selection of states which minimizes the retrieval sampling requirement. We stochastically train in the phase space of possible state selections, which is a bitstring of size $O(162)$ indicating whether each of the eligible states is present in the selection. We employ random bit-flips to train this bitstring, taking the smaller $\mathscr{L}$ at each step and artificially maintaining a sparse bitstring. We also occasionally branch into several lines and skip the cost function test, to avoid local minima. Finally, we consider variations $\mathscr{L}_i$ of the cost function, where the $i$ worst observables are not taken into account while computing the sampling requirement. This allows us to obtain state selections with different tolerance to error.

\subsection{Implementation results}

A target preferred parity was obtained from the first optimization, using scoring function $\mathcal{L}$. Then, using the state selection cost function $\mathscr{L}_i$ for a varying tolerance to error $0\le i\le 6$, we obtain state selections with sizes $9 \le N_\text{s}\le 14$. Afterwards, we compute performance metrics for encoding success probability and for retrieval sampling requirement, using negligible retrieval errors. The results are shown in Table \ref{tablepi4}.

\begin{table}
\begin{center}
  \begin{tabular}{ l  l  l  l  l }
    \Xhline{2\arrayrulewidth}
    \text{$N_\text{s}$\,\,\,\,\,\,} & \text{\,\,\,\,\,\textbf{$f^\text{e}$}\,\,\,\,\,\,\,} &\,$T$\,\,\,\,& \text{\,\,\,\,\,\,\,$\nu$}\,\,\,\,\,\,\,\,\,\,\, & \textbf{\text{\,$S$}}\,\,\,\, \\ \hline
    $9$ & $75/81$ &$11$& $0.0078$ & $35$ \\ 
    $11$ & $76/81$ &$12$& $0.0048$ & $44$ \\ 
    $12$ & $78/81$ &$13$& $0.0026$ & $49$ \\ 
    $13$ & $80/81$ &$14$& $0.0014$ & $55$ \\ 
    $14$ & $81/81$ &$14$& $0.0015$ & $59$ \\ \Xhline{2\arrayrulewidth}
  \end{tabular}
\end{center}
\vspace*{-0.3cm}
\caption{Performance metrics for a QRAC trained to encode $40$ binary digits of $\pi/4$ into $N_\text{s}$ context eigenstates of $4$ qubits, via the statistics of $81$ Pauli observables. We consider $9\le N_\text{s}\le 14$; $N_\text{s}=9$ is the least states that could cover all observables with well-defined parities (contexts of $n=4$ are size $9$), and $N_\text{s}=14$ is the smallest instance of perfect encoding. The main performance metrics are the encoding success probability $f^\text{e}$ and the retrieval sampling requirement $S$. To compute the latter, we consider a negligible retrieval noise $\nu$ (probability of not filtering an unbiased state) with respect to encoding errors by setting a number of retrieval steps $T$. The $N_\text{s}=14$ is an exception, as it already has perfect encoding, so we set the same $T$ as the previous case.} \label{tablepi4}
\end{table}

The metric $f^\text{e}$ is the fraction of observables which match the target preferred parity, and can be interpreted as the encoding success probability. The rest of the metrics are related to retrieval performance. We briefly explain them in the following. For retrieval in an arbitrary Pauli observable $\mathcal{O}$, the number of well-defined states is $N_\mathcal{O}\approx \mathcal{P}_\mathcal{O} N_\text{s}$, where $\mathcal{P}_\mathcal{O}=(2^{n-1}+1)/3^n=9/81$ is the fraction of observables in a given context. This way, $N_\mathcal{O}\approx N_\text{s}/9$. The retrieval consists in measuring the states on $\mathcal{O}$ during $T$ steps, filtering out those with unbiased statistics. The expected number of unbiased states which fail to be filtered is $\nu=(N_\text{s}-N_\mathcal{O})/2^{T-1}$, and the total number of states measured is $S\approx 3(N_\text{s}-N_\mathcal{O})+T N_\mathcal{O}$ (sampling requirement). We choose $T$ such that $\nu$ is negligible with respect to encoding errors, $\nu\ll (1-f^\text{e})$; specifically choosing $\nu \approx (1-f^\text{e})/10$ in this problem. An exception is made for the $N_\text{s}=14$ case, as its encoding is perfect, so the previous $T$ for $N_\text{s}=13$ is maintained. 

With this implementation, we have shown a simplified training process to store data into $n$-body Pauli statistics, allowing random access by measuring those observables. We were able to reach a perfect encoding for $40$ bits with as few as $14$ different states of $4$ qubits, or imperfect encodings with less states. However, in all cases, the total number of resources sent is greater than the original data, so for $n=4$ our proposal is technically not a QRAC. In the following section, we develop a statistical model to calculate the performance of our proposal for larger $n$.}

\section{Statistical extrapolation}
\label{sec:Statisticalextrapolation}

In this section, we measure the efficiency of this protocol via a statistical analysis of both the retrieval sampling requirement and the success probability. This is done for arbitrary number of qubits $n$, and in particular for even and large $n \gtrsim 10$ values.

\vspace{0.5cm}
\subsection{State selection size}
First, we estimate the number $N_\text{s}$ of different context eigenstates that Alice can use in the encoding, {in terms of how perfect she wants them to be. We consider $N_\mathcal{C}=2 \times 3^n$ contexts of size $d_{\mathcal{C}}=2^{n-1}+1$, each with $2^n$ unique context eigenstates. Then, we model the number of matches between arbitrary context eigenstates and arbitrary preferred parities using a binomial distribution
\begin{equation}
B(d_{\mathcal{C}}, 1/2 ; \cdot)\text{,}
\end{equation}
since each well-defined parity of a context eigenstate has probability $1/2$ of matching an arbitrary preferred parity. Then, we model the number of eigenstates with up to $h$ mismatches in their context as
\begin{equation}
B(2^n N_\mathcal{C}, p_h ; \cdot)\text{,}
\label{numbereigen}
\end{equation}
where $p_h=\sum_{k=0}^h B(d_{\mathcal{C}}, 1/2 ; k)$ is the probability that an eigenstate has up to $h$ mismatches, and it can be lower-bounded by $p_h'= { d_{\mathcal{C}} \choose{h}} 2^{-d_{\mathcal{C}}}$}. This defines the eligible states given some mismatch tolerance $\varepsilon=h/d_{\mathcal{C}}<1/2$. Then, we take into account the $2$-to-$1$ encoding, which causes $2^m$ configurations of preferred parities to be compatible with the data, where $m=(3^n-1)/2$. {Each compatible configuration is described by one instance of the statistical model in~\eqref{numbereigen}.} We compute $N_\text{s}$ as the largest number of eligible states (up to mismatch tolerance $\varepsilon$) between compatible configurations. {This can be done by finding the number $N_\text{s}$ such that the expected number of compatible configurations with that many eligible states is $1$, i.e.}
\begin{equation}
2^m B(2^n N_\mathcal{C}, p'_h ; N_\text{s}) = 1\text{.}
\end{equation}
Applying Stirling's approximation and keeping the dominant term in $n$ yields
\begin{equation}
N_\text{s}(n,\varepsilon)=g(\varepsilon)^{-1} \left(\frac{3}{2}\right)^n\text{,}
\label{ns}
\end{equation}
where $g(\varepsilon)=1-\frac{1-\varepsilon}{\ln(2)}\sum_{i=1}^\infty \frac{{\varepsilon}^i}{i}-\varepsilon \log_2(1/\varepsilon)$.
Notice that $N_\text{s}\to (3/2)^n$ when $\varepsilon\to 0$. {Also note that $N_\text{s}$ cannot be greater than $2\times 6^n$ (for $n\ge 4$), as we are considering $2\times 3^n$ contexts with $2^n$ eigenstates each. This imposes a restriction of $\varepsilon<0.474$, which softens as $n$ increases above $4$, but which poses no problem.}

\subsection{Sampling requirement}
Let us now compute the retrieval sampling requirement. Given a process to retrieve {the parity of an arbitrary observable $\mathcal{O}$ with $T$ measurement steps}, half of the unbiased-parity states are discarded at each step after the first one, so the total number of measurements performed on undefined-parity states is $(N_\text{s}-N_\mathcal{O})(1+\sum_{i=0}^{T-2} 2^{-i}) \approx 3(N_\text{s}-N_\mathcal{O})$, whereas the $N_\mathcal{O}$ well-defined states are never discarded and are measured $T$ times. Then, the total number of states used to measure one observable is
\begin{equation}
S(n,\varepsilon,N_\mathcal{O},T) \approx 3(N_\text{s}(n,\varepsilon)-N_\mathcal{O})+T N_\mathcal{O}\text{,}
  \label{Msamples}
\end{equation}
{though we can also upper-bound the sampling requirement with $T N_\text{s}$, considering that $T$ is the maximum number of copies needed per state. This bound has the advantage of being applicable to measurements with a context of observables instead of just a single one (commuting observables can share the same samples).}

We have yet to set an adequate value for $T$. In the process of measuring undefined-parity states, an average of $\nu=(N_\text{s}-N_\mathcal{O})/2^{T-1}$ states fail to be filtered out, which can be interpreted as the probability of having a single unfiltered state if $\nu \ll 1$, i.e. a retrieval noise probability. 
By inverting this relation, we determine $T=\log_2(N_\text{s}-N_\mathcal{O})-\log_2(\nu)+1$. Note that $\nu$ can be set to a negligible value with respect to encoding errors without increasing $S$ {(or $T N_\text{s}$)} too much. In that case, virtually no unbiased states are mistakenly labeled as well-defined states while computing the preferred parity of $\mathcal{O}$. This way, the overall success probability depends only on encoding errors, and not on any undefined state that failed to be filtered during retrieval.

\subsection{Success probability} 
{The value of $1-\varepsilon$ can be used to lower-bound the probability of one parity matching between one state of the selection and the desired preferred parity. We use it to compute the encoding success probability of our set of $N_\text{s}$ states. For this end,} let $\mathcal{O}$ be an arbitrary $n$-body Pauli observable whose preferred parity we want to retrieve from the set of states. Statistically, most of the selected states have an unbiased parity in this observable. As noted before, unbiased states are later filtered out by the retrieval protocol, so only well-defined parities matter for a successful encoding. {Given $N_\text{s}$ random context eigenstates, we say that $N_\mathcal{O}$ are well-defined in $\mathcal{O}$. Considering that $\mathcal{P}_\mathcal{O}=(2^{n-1}+1)/3^n\approx (1/2)(2/3)^n$ is the fraction of observables contained in any context, $N_\mathcal{O}$ can be modeled by $B(N_\text{s}, \mathcal{P}_\mathcal{O};\cdot)$, and from Eq.~\eqref{ns} we note that $\langle N_\mathcal{O} \rangle=\mathcal{P}_\mathcal{O} N_\text{s} \approx (1/2) g(\varepsilon)^{-1}$, which does not depend on $n$.} If the majority of the $N_\mathcal{O}$ well-defined states yield the target parity, the encoding has been successful. This way, we compute the encoding success probability $f^\text{e}$ as a function of $n$ and $\varepsilon$: 

\begin{equation}
f^\text{e}=
\sum_{i= \lfloor N_\mathcal{O}/2 \rfloor+1}^{N_\mathcal{O}} B(N_\mathcal{O},1-\varepsilon;\,i) ,
\label{rretrieval}
\end{equation} 

\noindent for odd $N_\mathcal{O}$, {with an additional term $B(N_\mathcal{O},1-\varepsilon;\,N_\mathcal{O}/2)/2$ in even cases. 

Notably, since $\langle N_\mathcal{O} \rangle \approx (1/2) g(\varepsilon)^{-1}$ is independent of $n$, it seems intuitive that the success probability is also independent of $n$. To be certain, we numerically compute the expected success probability $\langle f^\text{e} \rangle$ for encoding one preferred parity using $N_\mathcal{O} \sim B(N_\text{s}, \mathcal{P}_\mathcal{O};\cdot)$. Since a data bit is retrieved successfully when both of its associated parities are correct or both are wrong, the expected success probability for encoding each data bit is $\langle \bar{f}^\text{e} \rangle=(\langle f^\text{e} \rangle)^2+(1-\langle f^\text{e} \rangle)^2$. This expression does not depend on the variance of $\bar{f}^\text{e}$ due to $N_\mathcal{O}$ being mostly independent for each observable.}

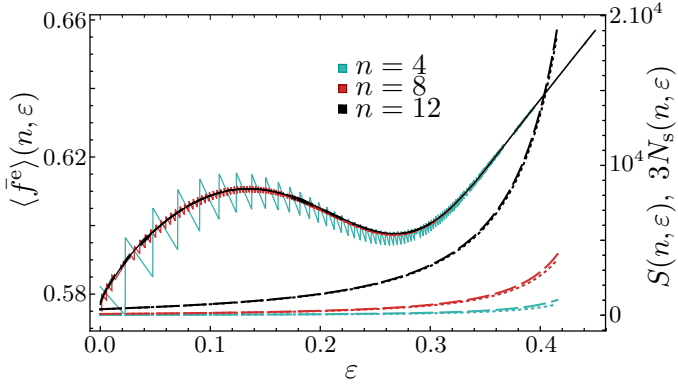
\begin{figure}

\centering
\hspace*{-0.3cm}
\includegraphics[width=.41\textwidth]{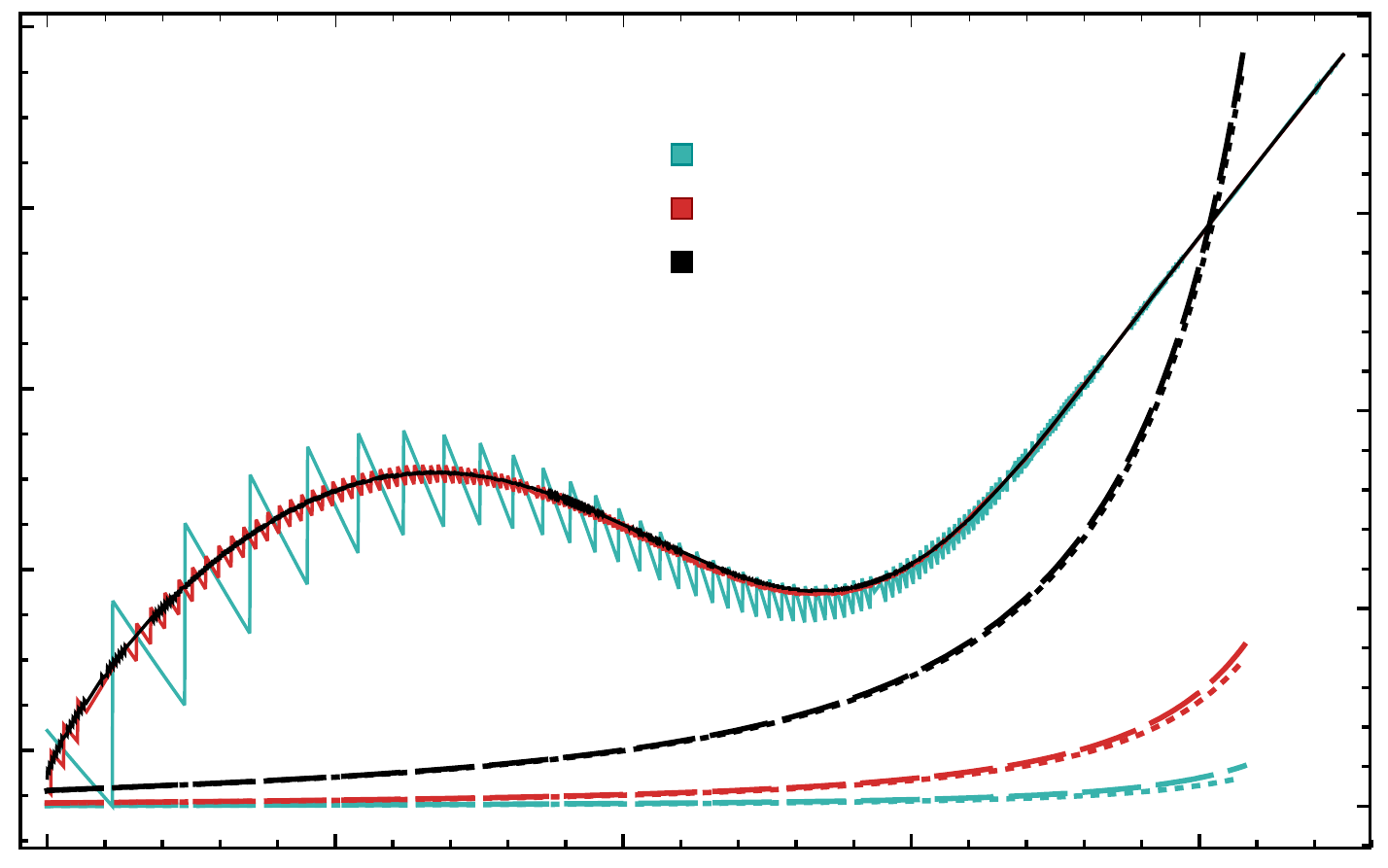}\hfill
\put(-227,49){\begin{sideways} $\langle \bar{f}^\text{e} \rangle(n,\varepsilon)$ \end{sideways}}
\put(13,19){\begin{sideways} $S(n,\varepsilon),~3 N_\text{s}(n,\varepsilon)$ \end{sideways}}
\put(-96,99.8){$n=4$}
\put(-96,92.0){$n=8$}
\put(-96,83.5){$n=12$}
\put(-197.55,-5){{\footnotesize $0.0$}}
\put(-156.25,-5){{\footnotesize $0.1$}}
\put(-114.95,-5){{\footnotesize $0.2$}}
\put(-73.65,-5){{\footnotesize $0.3$}}
\put(-32.35,-5){{\footnotesize $0.4$}}
\put(-100,-15){$\varepsilon$}
\put(-1,118.9){{\footnotesize $2.10^4$}}
\put(-1,63){{\footnotesize $10^4$}}
\put(-1,6.3){{\footnotesize $0$}}
\put(-213,117){{\footnotesize $0.66$}}
\put(-213,65.6){{\footnotesize $0.62$}}
\put(-213,14){{\footnotesize $0.58$}}

\caption{Left axis: Expected success probability $\langle \bar{f}^\text{e} \rangle$ for encoding a data bit (solid line) for $n=4,8,12$ (cyan, red, and black, respectively), as a function of the encoding mismatch tolerance $\varepsilon$. {Notably, there is an asymptotic $\langle \bar{f}^\text{e} \rangle (\varepsilon)$ curve as $n$ increases.} Right axis: number of required samples $S$ (dashed line) for negligible retrieval noise $\nu=2^{-6} \ll (1-\langle f^\text{e} \rangle)$, and eigenstate selection size $N_\text{s}$ multiplied by $3$ (dotted line), for same system sizes and as a function of $\varepsilon$. Notice that $S$ and $3 N_\text{s}$ are approximately equal, implying that in average $3$ measurements are required per eigenstate (see Eq. \eqref{Msamples}).}
\label{dualplots}
\end{figure}

{The result of this computation is shown in Fig.~\ref{dualplots}, where we provide $\langle \bar{f}^\text{e} \rangle$ as a function of $\varepsilon$, for $n=4,8,12$.} Discontinuities in $\langle \bar{f}^\text{e} \rangle$ tend to vanish for large number of selected states $N_\text{s}$, such that it converges into a {single smooth curve} for large number of qubits ($n \gtrsim 8$) and large mismatch tolerances ($\varepsilon \gtrsim 0.3$), as they both increase the number of selected eigenstates. {Most importantly, for large $n$, the success probability reaches an asymptotic value defined only by $\varepsilon$ and higher than $0.5$. This allows to reach any success probability with a fixed number of repetitions, as we show in the next subsection.}
On a second note, when increasing $\varepsilon$, the increase of the number of selected eigenstates favors higher success probability but at the same time hinders it because {the states have a poorer resemblance to the target preferred parities.} 
This dual behaviour explains the local minima of $\langle \bar{f}^\text{e} \rangle$.

Notice that the local maxima of $\langle \bar{f}^\text{e} \rangle$ is not necessarily the optimal $\varepsilon$, as the sampling requirement must be taken into account too. {If we want a fair comparison between $\varepsilon$ values, we need to fix either the success probability or the sampling requirement and compare the other value in terms of $\varepsilon$. We do this by introducing repetitions of the code.}

\subsection{Repetitions} In order to obtain the mismatch tolerance $\varepsilon$ that minimizes the sampling requirement given a fixed success probability, we increase the success probability by means of protocol repetitions with majority-rule, fixing it to a constant value among all $\varepsilon$. This allows us to compare sampling requirements (modified by the repetitions) and calculate the optimal mismatch tolerance $\varepsilon^\star$. In general, a code with success probability $\bar{f}$ repeated $r$ times yields an amplified success probability
\begin{equation}
\mathcal{F}=\sum_{j=(r+1)/2}^r B(r,\bar{f};j)\text{,}
  \label{amplification}
\end{equation}
assuming odd $r$ for simplicity.
Thus, minimizing the sampling requirement for fixed $\mathcal{F}$ (through $r$ repetitions) yields an optimal mismatch tolerance ${\varepsilon}^\star=0.0480$ ($\pm 0.0034$) for $n\ge 16$.

{For $n\ge 16$ and using the optimal $\varepsilon$, the average encoding success probability in our code is $\langle \bar{f}^\text{e} \rangle (\varepsilon^\star)=0.5983$, and also acts as general success probability, due to the negligible retrieval noise $\nu \ll (1- f^\text{e})$. Considering that this success probability is low but independent of $n$, we note that we can apply protocol repetitions not just as a mathematical artifice to optimize $\varepsilon$, but also as an actual tool for the protocol. Alice could use a different $\mathcal{G}$ ordering on each iteration to guarantee statistical independence for the encodings, and Bob would apply a majority-rule between the iterations to retrieve each data bit. This way, for $n\ge 16$, $r=243$ repetitions suffice to achieve final success probability $\mathcal{F}=0.999$. Furthermore, to obtain an asymptotic relationship between $r$ and $\mathcal{F}$, the former can be isolated from Eq. \eqref{amplification} by taking the normal approximation. Using the success probability $\langle \bar{f}^\text{e} \rangle (\varepsilon^\star)$, this yields $r \approx 49.7\, (\text{Erfc}^{-1}(2-2 \mathcal{F}))^2$, where $\text{Erfc}^{-1}(\cdot)$ is the inverse complementary error function. Considering that $\text{Erfc}^{-1}(y)\approx\sqrt{\ln(1/y)}$ for $y \ll 1$, we obtain $r\sim O(\ln(1/(1-\mathcal{F})))$, i.e. infinite repetitions for the limit $\mathcal{F}\rightarrow 1^{-}$, as is to be expected.}

{Now that the final success probability $\mathcal{F}$ is defined, we substitute $\varepsilon^\star$ to re-calculate all parameters related to the sampling requirement. With $\varepsilon^\star$, the size of the state selection becomes
\begin{equation}
N_\text{s}(n)=g(\varepsilon^\star)^{-1} \left(\frac{3}{2}\right)^n=1.385 \left(\frac{3}{2}\right)^n \text{.}
  \label{Nstar}
\end{equation}}
Then, with $\varepsilon^\star$, the number of retrieval steps can be upper-bounded by 
\begin{equation}
T(n)=n \log_2 (3/2) + 9 \text{,}
  \label{Tretrieval}
\end{equation}
{using $\nu= 2^{-7.5}\ll (1-f^\text{e})$, and considering $\log_2(N_\text{s}-N_\mathcal{O})\approx \log_2(N_\text{s})$ due to $\langle N_\mathcal{O} \rangle= 0.6925 \ll N_\text{s}$}.
Recall that $N_\text{s} T$ bounds the total number of states used in the code, even if a full context of observables is measured in the retrieval.

Equivalently, with $\varepsilon^\star$, we can express the expected sampling requirement for retrieval of one observable (Eq. \eqref{Msamples}) as a function of $n$ alone,
{\begin{equation}
S(n) =4.155 \left(\frac{3}{2}\right)^n + n \log_2(3/2)+6\text{,}
  \label{Moptimal}
\end{equation}
where we have rounded up $N_\mathcal{O}=1$ (from $\langle N_\mathcal{O} \rangle= 0.6925$) to be on the safe side.

\vspace*{0.5cm}
\section{Wrapping up the QRAC}
\label{sec:WrappinguptheQRAC}

This way, we can summarize the whole $m \rightarrow k$ QRAC, including the repetitions, as a function of $n$ and the success probability $\mathcal{F}$. Alice has a data of 
\begin{equation}
m=(3^n-1)/2 \sim O(3^n)
\end{equation}
bits that she wants Bob to randomly access (one context at a time). To grant this access she sends at most $N_\text{s}(n) T(n) \sim O(n (3/2)^n)$ states of $n$ qubits, and repeats this $r\sim O(\ln(1/(1-\mathcal{F})))$ times to reach a final success probability of $\mathcal{F}$. At most, the total number of states sent is $N=r N_\text{s} T$ and the total number of qubits sent is $n$ times that:
\begin{equation}
k= r n N_\text{s}(n) T(n) \sim O\left(n^2 \left(\frac{3}{2}\right)^n \ln\left(\tfrac{1}{1-\mathcal{F}}\right)\right)\text{.}
\end{equation}
Here, the final success probability $\mathcal{F}$ imposes a fixed number of repetitions, but it does not scale with $n$, thus the number of sent resources $k$ scales with a lower order than the original data size $m$. We can choose an arbitrarily-high success probability $\mathcal{F}$ and show that the compression ratio $k/m$ of this QRAC is of order $O(n^2/2^n)$, or $O(m^{-0.63})$ in terms of the original data size. This way, setting $r=243$ for success probability $\mathcal{F}=0.999$, Alice sends less resources than her original data for $n \ge 18$, achieving compression. If instead we set no repetitions, the success probability is $\langle \bar{f}^\text{e} \rangle\approx 0.5983$ and compression is achieved for $n \ge 10$.}

\section{Comparison}
\label{sec:Comparison}

{To compare our $m\rightarrow k$ QRAC against other codes, we consider the protocol before the repetitions, for simplicity. This means that the number of resources sent by Alice is $k=n N_s(n) T(n) \sim O(n^2 (3/2)^n)$ and the success probability is $\langle \bar{f}^\text{e} \rangle =0.5983$. Like before, $m=(3^n-1)/2$ and the compression ratio is $k/m\sim O(n^2 /2^n)$. As mentioned earlier, we can assume the data encoded in the code to already be in the lossless compression limit of maximum entropy. This way, any additional lossless compression is quantum advantage, and any additional lossy compression with greater success probability than classical RACs is also quantum advantage.

\subsection{Comparison against RACs}

To perform this comparison, we need to calculate the success probability of an equivalent $m\rightarrow k$ random access code. This can be done partitioning the problem into $k$ RACs of the form $\tfrac{m}{k}\rightarrow~1$. Each one would send $1$ bit and grant random access to $m/k$ bits, with average success probability greater than $1/2$. The original data in each of these smaller RACs does not overlap, so it is sufficient to calculate the success probability of a single one of them. This has previously been calculated in \cite{ambainis2008quantum}. In the following, we obtain the same result with a slightly different procedure, specifically tailored to our problem.

To compute the success probability of a $\tfrac{m}{k}\rightarrow~1$ RAC we use the \textit{most-common-bit} approach which was used in the $3\rightarrow 1$ RAC. Essentially, the single communicated bit indicates what is the most common bit in the original data, and this information allows us to use that bit as guess for any of the $m/k$ bits. We can model the total number of $1$s in an arbitrary $\tfrac{m}{k}$-bit configuration with a binomial distribution $B(m/k,1/2;\cdot)$. If the communicated bit is $0$($1$), it means that our current configuration lies in the first(second) half of the distribution (assume odd $m/k$ for simplicity). This way, the expected number of $0$s($1$s) in the configuration is obtained from the expectation of the renormalized first(second) half of the distribution. Dividing this result by the bitstring length $m/k$, we obtain the probability of an arbitrary bit in the configuration being $0$($1$), that is, the success probability $f_\text{RAC}^{m/k\rightarrow 1}$. Its value is the same for either communicated bit $0$ or $1$:
\begin{equation}
f_\text{RAC}^{m/k\rightarrow 1}= \left[ \sum_{i= \lfloor (m/k)/2 \rfloor+1}^{m/k}2 \, i \, B(m/k,1/2;\,i) \right] /(\tfrac{m}{k}) \text{.}
\end{equation}
Note that this yields the known success probability $f_\text{RAC}^{m/k\rightarrow 1}=0.75$ for $m/k=3$.

For large $m/k$ we can obtain a cleaner expression with the normal approximation.
\begin{equation}
f_\text{RAC}^{m/k\rightarrow 1}= \left[ \int_{\tfrac{m/k}{2}}^{\infty} 2 x \, N \left(\tfrac{m/k}{2},\, \tfrac{\sqrt{m/k}}{2}; \, x \right) dx \right] /(\tfrac{m}{k}) \text{,}
\end{equation}
where $N(\mu,\sigma;\, \cdot)$ is the normal distribution with mean $\mu$ and standard deviation $\sigma$. When solved, this yields
\begin{equation}
f_\text{RAC}^{m/k\rightarrow 1}= \frac{1}{2}+\sqrt{\frac{k/m}{2 \pi}}\text{,}
\label{eqfrac}
\end{equation}
in agreement with \cite{ambainis2008quantum}.

Our QRAC has asymptotic success probability (per bit) $\langle \bar{f}^\text{e} \rangle=0.5983$. It is greater than RAC success probability in Eq. \eqref{eqfrac} for $k/m<0.0607$, that is, $n \ge 14$ when substituting $k$ and $m$ in terms of $n$. This is because RAC success probability decays exponentially fast in $n$, with $f_\text{RAC}^{m/k\rightarrow 1}=1/2+O\left(\tfrac{n}{2^{n/2}}\right)$. In short, our QRACs are better than RACs when employing systems of $n \ge 14$ qubits.


\subsection{Comparison against QRACs}

In contrast to a RAC, a $m \rightarrow k$ QRAC cannot be partitioned into $k$ QRACs of the form $\tfrac{m}{k} \rightarrow 1$ without losing significant quantum advantage. This is because QRACs are fundamentally defined by a number of measurement-basis choices, which in general grow exponentially with the number of qubits of the sent state(s). This is both true for MUBs or $n$-body Pauli observables.
Ambainis et. al. \cite{ambainis2008quantum} have shown lower and upper bounds for the success probability of $m' \rightarrow 1$ QRACs. Setting $m'=m/k$ for ease of comparison, these bounds are: 

\begin{equation}
\frac{1}{2}+\sqrt{\frac{k/m}{(3/2) \pi}} \le f_\text{QRAC}^{m/k\rightarrow 1} \le \frac{1}{2}+\frac{\sqrt{k/m}}{2}\text{,}
\label{qracbound}
\end{equation}
where the lower bound is greater than $f_\text{RAC}^{m/k\rightarrow 1}$. Furthermore, building a $m\rightarrow k$ QRAC with $k$ QRACs of the form $\tfrac{m}{k}\rightarrow~1$ shows that our QRAC proposal has greater success probability than this upper bound when $n \ge 16$. This further proves that partitioning a QRAC into smaller pieces is suboptimal.

\begin{table}
\centering
\begin{tabular}{ | c || c | c |  }
 \hline
 \multicolumn{3}{|c|}{\raisebox{-.5\height}{\textbf{QRAC} ${(d+1)\log_2(d) \rightarrow \log_2(d)}$}} \\[5pt]
 \hline
 \raisebox{-.5\height}{\;\;\;\;\;$d$\;\;\;\;\;}& \raisebox{-.4\height}{\;\;\;\;\;\;\;$f_\text{dit}$\;\;\;\;\;\;\;}&\raisebox{-.4\height}{\;\;\;\;\;\;\;$f_\text{bit}$\;\;\;\;\;\;\;}\\[5pt]
 \hline
 $2$&$0.789$&$0.789$\\
 $3$&$0.637$& $0.7524$\\
 $4$&$0.5424$&$0.7365$\\
 $5$&$0.4700$&$0.7224$\\
 $7$&$0.3720$&$0.7031$\\
 $8$&$0.3372$&$0.6960$\\
 \hline
\end{tabular}\par
\caption{Success probability $f_\text{dit}$ for $d$-dimensional QRACs, of the form ${(d+1)\log_2(d) \rightarrow \log_2(d)}$ (in $\text{bit}\!\!\rightarrow\!\!\text{qubit}$ notation), as obtained in \cite{casaccino2008extrema} for all prime-power $d\le 8$. Since the original success probability is calculated for guessing dits ($d$-levels), we additionally compute a metric for bit success probability, $f_\text{bit}=(f_\text{dit})^{1/\log_2(d)}$.}
\label{MUBQRAC}
\end{table}

Thus, the success probability $m\rightarrow k$ QRACs with high-dimensional quantum systems is still unbounded, and this is precisely the family of QRACs which can exploit the exponential increase of measurement-basis options. To the best of our knowledge, the best precedent we have of high-dimensional QRACs in the literature is the $(d+1)\log_2(d) \rightarrow \log_2(d)$ QRAC of Casaccino et. al. \cite{casaccino2008extrema}, which employs MUBs and has been computed for $2 \le d \le 8$ (skipping $d=6$ as it is not a prime-power). They show decaying success probabilities as $d$ increases, but these probabilities are for retrieving a full dit, which has zero-knowledge success probability of $1/d$, something that already decays in $d$. For a fair comparison of success  probability between different values of $d$ (and against all other RACs and QRACs we have mentioned so far), we use their dit success probability $f_\text{dit}$ to estimate a bit success probability $f_\text{bit}=(f_\text{dit})^{1/\log_2(d)}$. As we can see, this function transforms zero-knowledge success probability ($d$-level) to zero-knowledge success probability ($2$-level) $f_\text{dit}=1/d \rightarrow f_\text{bit}=(1/d)^{\log_2(d)}=1/2$ for any $d\ge 2$ and also transforms $f_\text{dit}=1 \rightarrow f_\text{bit}=1$. We apply it on the results of their QRAC (Table \ref{MUBQRAC}), and note that an asymptotic behaviour of bit success probability greater than $1/2$ cannot be discarded until further study. We also note that the bit success probabilities on the $d=4$ and $d=8$ cases, which correspond to $10 \rightarrow 2$ and $27 \rightarrow 3$ QRACs, slightly breaks the upper bound in Eq. \eqref{qracbound} for $5 \rightarrow 1$ and $9 \rightarrow 1$ QRACs, respectively. Once again, this shows that higher dimensional QRACs provide an improvement, and that partitioning QRACs (e.g. partitioning a $27 \rightarrow 3$ QRAC into three $9 \rightarrow 1$ QRACs) reduces their advantage.

These MUB-QRACs showcase a compression ratio of $1/(d+1)\sim O(1/2^n)$, exponential in $n$ when considering $n$-qubit systems, like our proposed QRAC. For now, our approach has the advantage of having a confirmed asymptotic bit success probability (greater than $1/2$), and of encoding the data into states which can be produced with low-depth quantum circuits (Fig. \ref{maincircuit}). Finally, the mathematical complexity of working with MUBs for large $n$ is a drawback to also take into account ($d=8$ is $3$ qubits, and it already is hard), whereas $n$-body Paulis are readily-defined, and their context construction is provided for any $n$ in this paper.


\section{Applications}
\label{sec:Applications}

Besides its large compression ratio (small $k/m$), the $m \rightarrow k$ QRAC presented in this paper has a distinguishing feature from previous QRACs: its success probability can be amplified (by repetitions) to near-perfect values while still being a random access code, that is, while encoding the original data into \textit{less} information units.

Where's the catch? The amount of information that is randomly accessed by Bob must be much smaller than the original data, even if it can be any fragment of the data. This is certainly a limitation, but leaves room for applications where a large amount of data needs to be stored and only-partially queried. For instance, as exemplified in \cite{ambainis1999dense}, one could encode the contents of an entire telephone directory and then, via a suitably chosen measurement, look up any single telephone number of choice. Or, using a more modern example, encode everyone's Google Drive data and, via suitable measurements, retrieve any single file. Indeed, Google Drive has approximately $1$ billion users \cite{googledrivesize} with $\sim 100$ Gigabytes each ($800$ Giga-bits), and the amount of data that our QRAC can randomly access is $m=(3^n-1)$ bits, in terms of the size $n$ of quantum states used. This way, systems of $n=44$ qubits are sufficient to randomly access such volumes of information, with a compression ratio $k/m=5.75 \times 10^{-8}$ for near-perfect bit success probability $\mathcal{F}=0.999$.

Our proposed QRAC is also a fitting solution to store and consult decision trees, i.e. brute-force strategies. The total size of a decision tree scales exponentially fast, but only one branch needs to be consulted at a time if we want to put it to use. For instance, a brute-force solution for chess would imply storing a database with what move to make on each possible board configuration. Chess has an upper bound of $2^{155}$ existing board configurations~\cite{tromp2010john} and a branching factor of $\sim35$ possible plays per turn~\cite{levinovitz2014mystery}, thus there exist $35^{2^{155}}$ brute-force strategies (all possible moves for all possible board configurations). Storing one of them amounts to $2.34 \times 10^{47}$ bits. Our QRAC would need systems of $n=100$ qubits to store this volume of information, and would reach a compression ratio of $k/m=3.52 \times 10^{-24}$ for near-perfect bit success probability $\mathcal{F}=0.999$.

\section{Conclusion and outlook}
\label{sec:Conclusionandoutlook}

We have proposed a $m \rightarrow k$ Quantum Random Access Code (QRAC) which grants random access to an $m$-bit data by sending relatively few quantum states of high-dimension ($n$-qubit states). In terms of $n$, the original data has $m\sim O(3^n)$ bits and the total number of qubits sent is $k\sim O(n^2 (3/2)^n)$ for an arbitrarily high bit-retrieval success probability $\mathcal{F}$, or $k\sim O(n^2 (3/2)^n \ln(1/(1-\mathcal{F})))$ in terms of it. This means that the compression ratio between sent resources and original data is $k/m\sim O(n^2/2^n)$, or $O(m^{-0.63})$ in terms of the original data size $m$. This is the first instance of a QRAC being calculated on the asymptotic limit of large quantum systems.

Instead of using Mutually Unbiased Bases (MUBs) like previous proposals, we based our QRAC on the family of $n$-body Pauli observables grouped into different measurement contexts, and encoded our data using the most convenient eigenstates to these contexts. We showed that these states can be generated with linear-depth quantum circuits, and provided a toy implementation on $n=4$, storing $40$ binary digits of $\pi/4$ into quantum states. Notably, the current main limitation to large-$n$ implementation is not the quantum hardware, but classical training routines, which we have implemented simplistically.

Contrasting against previous literature, our QRAC has no advantage for low number of qubits, but improves upon any other random access code for large $n$. We obtain greater success probability than classical random access codes for $n\ge 14$, greater success probability than other quantum random access codes for $n\ge 16$, and can amplify the success probability to near-perfect values ($\mathcal{F}=0.999$) for $n\ge 18$ without ceasing to be a QRAC (less resources sent than held originally).

Finally, we note on the applicability of our QRAC to problems where large data is to be stored and accessed with small queries. Everyone's Google Drive data, equivalent to $m\sim O(10^{20})$ bits, could be stored using systems of $n=44$ qubits and one file (any file) be retrieved using $O(10^{-8})$ times less resources than $m$. Also, decision trees are fitting problems for our QRAC, as the trees can be very large but require querying for a single branch. We take as example a brute-force solution for chess, equivalent to $m\sim O(10^{47})$ bits of data, which can be stored using systems of $n=100$ systems. A game of chess could be optimally played using $O(10^{-24})$ times less resources than $m$.

In future research lines, classical optimization routines to find the best context eigenstate for the encoding should be improved in order to scale up the implementation from a toy model of four-qubit systems to the relevant $n\ge 14$ regime in which Contextual QRACs show advantage.
}

\textit{Acknowledgements.--} 
The authors acknowledge financial support from the Basque Government through Grant No. IT1470-22 and from the Basque Government QUANTEK project under the ELKARTEK program (KK-2021/00070), the Spanish Ram\'{o}n y Cajal Grant No. RYC-2020-030503-I, and project Grant No. PID2021-125823NA-I00 funded by MCIN/AEI/10.13039/501100011033 and by ``ERDF A way of making Europe'' and ``ERDF Invest in your Future'', as well as from OpenSuperQ (820363) of the EU Flagship on Quantum Technologies, and the EU FET-Open projects Quromorphic (828826) and EPIQUS (899368). 




%

\end{document}